%% file: unc_proofs_advice.tex
\documentclass{llncs}
\usepackage{my-llncs}
\usepackage{braket}
\usepackage{bbm}
\usepackage{enumitem}

\usepackage{tikz}
\usetikzlibrary{fit}

\newcommand{\Tr}{\mathrm{Tr}}

\newcommand{\state}[1]{\ket{#1}\!\!\bra{#1}}
\newcommand{\brakett}[2]{\langle #1| #2 \rangle}

\newcommand{\clA}{\mathcal{A}}
\newcommand{\clB}{\mathcal{B}}
\newcommand{\clC}{\mathcal{C}}
\newcommand{\clD}{\mathcal{D}}

\newcommand{\clO}{\mathcal{O}}
\newcommand{\clP}{\mathcal{P}}
\newcommand{\clQ}{\mathcal{Q}}
\newcommand{\clR}{\mathcal{R}}

\newcommand{\clX}{\mathcal{X}}
\newcommand{\clY}{\mathcal{Y}}
\newcommand{\clZ}{\mathcal{Z}}

\newcommand{\F}{\mathrm{F}}

\newcommand{\Id}{\mathbbm{1}}
\DeclareMathOperator*{\bbE}{\mathbb{E}}
\newcommand{\bbN}{\mathbb{N}}

\newcommand{\eps}{\varepsilon}
\newcommand{\supp}{\mathrm{supp}}

\newcommand{\QMA}{\mathrm{QMA}}
\newcommand{\cQMA}{\mathrm{cloneableQMA}}

\newcommand{\suQMA}{\mathrm{suQMA}}

\newcommand{\BQPq}{\mathrm{BQP}/\mathrm{qpoly}}

\newcommand{\BQPuq}{\mathrm{avBQP}/\mathrm{upoly}}
\newcommand{\BQPsuq}{\mathrm{avBQP}/\mathrm{supoly}}
\newcommand{\FBQPq}{\mathrm{FBQP}/\mathrm{qpoly}}

\newcommand{\FEQPq}{\mathrm{FEQP}/\mathrm{qpoly}}
\newcommand{\FEQPsuq}{\mathrm{FEQP}/\mathrm{supoly}}

\newcommand{\HMn}{\textsc{HM}_n}

\newcommand{\QOR}{\textsc{QuantumOr}}

\DeclareMathOperator{\D}{D}
\DeclareMathOperator{\R}{R}
\DeclareMathOperator{\Q}{Q}

\input{Macros}

\title{Are uncloneable proof and advice states strictly necessary?}
\author{Rohit Chatterjee\inst{1} \and Srijita Kundu\inst{2} \and Supartha Podder\inst{3}}
\institute{National University of Singapore \and Institute for Quantum Computing, University of Waterloo \and Stony Brook University}

\begin{document}
\maketitle
\begin{abstract}
Yes, we show that they are.

We initiate the study of languages that necessarily need uncloneable quantum proofs and advice. We define strictly uncloneable versions of the classes QMA, BQP/qpoly and FEQP/qpoly (which is the class of relational problems solvable exactly with polynomial-sized quantum advice). Strictly uncloneable QMA is defined to be the class of languages in QMA that \emph{only} have uncloneable proofs, i.e., given any family of candidate proof states, a polynomial-time cloning algorithm cannot act on it to produce states that are jointly usable by $k$ separate polynomial-time verifiers, for arbitrary polynomial $k$. This is a stronger notion of uncloneable proofs and advice than those considered in previous works, which only required the existence of a single family of proof or advice states that are uncloneable. We show that in the quantum oracle model, there exist languages in strictly uncloneable QMA and strictly uncloneable BQP/qpoly. The language in strictly uncloneable QMA also gives a quantum oracle separation between QMA and the class cloneableQMA introduced by Nehoran and Zhandry (2024). We also show \emph{without using any oracles} that the language, used by Aaronson, Buhrman and Kretschmer (2024) to separate FEQP/qpoly and FBQP/poly, is in strictly uncloneable FEQP/qpoly.
\end{abstract}

\section{Introduction}

The phenomenon of {\em uncloneability} of quantum states has emerged in recent years as an exciting new area of study in quantum information theory and quantum complexity. This originated in the foundational no-cloning principle \cite{Die82,WZ82}, and even early on this was seen to enable interesting applications, such as the ideas of quantum money and quantum key distribution \cite{Wie83, BB84}. 

Since then, there have been a wide variety of results attempting to harness uncloneability for exciting applications in quantum computing and quantum cryptography --- including various kinds of quantum money \cite{AC13,JOC:Zhandry21,STOC:Shmueli22,DBLP:KSS21}, and a plethora of uncloneable cryptographic primitives such as uncloneable encryption \cite{QIC:Gottesman03,TQC:BL20,CRYPTO:AKLLZ22,KT2025deviceindependent}, uncloneable decryption schemes \cite{EP:GZ20}, one-shot signature schemes \cite{STOC:AGKZ20}, quantum copy protection \cite{CCC:Aar09,Q:CMP20,C:ALLZZ21,AKL23,TCC:CHV23}, secure software leasing \cite{EC:AP21,TCC:BJLPS21}, and uncloneable zero-knowledge proofs \cite{TCC:GMR24,AC:JK24}. In addition, there has been the advent of certain notions that by design rely quite closely on some behaviour enabled by uncloneability --- these include ideas such as certified deletion \cite{TCC:BI20,C:BK22,EC:BGKMRR24,C:BR24} and revocable cryptography \cite{TCC:APV23,EC:AKNYY23,TQC:MPY24}. Broadly, these works try and leverage a certain sort of ephemerality enabled by families of carefully constructed uncloneable states that carry certain information, and the idea is that such information disappears along with the states when these are tampered with or subjected to forced replication. The novelty of such objects lies in the uniquely quantum phenomena of uncloneability --- any classical data can always be copied and the notion is vacuous in classical models of computing.   


The complexity theoretic implications of uncloneability, in contrast, have only recently begun to be studied. This direction of research was first suggested by \cite{blogpost}, who raised the question of which quantum proof and advice states could be made cloneable or uncloneable.

\para{Quantum proofs.}  The class QMA serves as the quantum analog of the class NP, consisting of a quantum polynomial-time verifier that gets a proof (also called witness) state from a prover and then runs a quantum algorithm to determine membership in the language. A big open question concerning this class lies in determining the contribution of the witness state to the power of this class. In particular, it is of significant interest to compare QMA with the class QCMA, which is the subclass of QMA with classical witnesses. Whether or not QMA is equal to QCMA is one of the big open problems of quantum complexity theory. Although it may be beyond our current capabilities to show an unconditional separation between the classes, quantum oracle separations between the two have been known, as first shown by Aaronson and Kuperberg \cite{AK07}. The notion of uncloneability comes up naturally in this context --- if a problem in QMA is shown to admit QMA proofs that are inherently uncloneable, it will let us separate QMA from QCMA, since problems in the latter class always admit classical QMA proofs that can be trivially cloned. 

Following Aaronson's question, some works have studied languages in QMA that have cloneable or uncloneable proofs. \cite{ITCS:NZ24} studied the subclass of QMA which consists of problems that have some cloneable proof (which they called cloneableQMA) --- this obviously contains QCMA (since all classical proofs are cloneable), but they showed that cloneableQMA can also be separated from QCMA with respect to a quantum oracle. Essentially, they demonstrated a quantum oracle problem that has proofs that are cloneable but not teleportable, which achieves the separation since classical data is trivially teleportable. Very recently, \cite{BGPS24} introduced the notion of an anti-piracy proof system, and demonstrated that in the oracle setting there is a language with an anti-piracy proof system. Their definition of anti-piracy involves the \emph{functionality} of a state as a proof resisting cloning, rather than the state itself, which is a stronger notion of uncloneability of proofs.

\para{Quantum advice.}
Closely related to the question of the relative power of quantum vs.\ classical proofs is the question of the relative power of quantum vs.\ classical advice. Advice differs from proofs in that it only depends on the input length and not other specifics of the input, and is always trusted. The class BQP/poly is defined as the class of decision problems solvable by BQP algorithms with polynomial-sized non-uniform classical advice. BQP/qpoly is defined similarly as BQP/poly, except with polynomial-sized quantum advice. \cite{AK07} gave a quantum oracle separation between BQP/poly and BQP/qpoly, with a construction very similar to their QMA vs.\ QCMA separation.

Like QMA vs.\  QCMA, we do not believe separating BQP/qpoly and BQP/poly unconditionally to be within our current capabilities, but separating the relational versions (FBQP/poly and FBQP/qpoly) of these classes turns out to be much easier. \cite{ABK23} used a separation between one-way quantum and classical communication complexity for a relational problem \cite{BJK08} to show that FBQP/poly and FBQP/qpoly are separated. The relational problem, called the Hidden Matching Problem, they considered is in fact in the exact or zero-error version of FBQP/qpoly (called FEQP/qpoly), and not in the bounded error FBQP/poly.

Like uncloneable proofs, languages that admit uncloneable quantum advice have previously been studied in \cite{BKL23}. The authors define a class, that they called BQP/upoly, which consists of languages in BQP/qpoly that admit uncloneable advice. Like \cite{BGPS24}, the notion of uncloneability here involves the functionality of the advice being uncloneable. They showed that if one is willing to assume the existence of copy protection schemes for pseudorandom functions, then there is a language in this class. Further, without using any assumptions, they are able to construct a promise problem that is in the promise version of BQP/upoly. However the languages that they consider are in P/poly, so they in fact also admit cloneable proofs.

\para{Notions of uncloneability.}  In most recent works on the topic, uncloneability is defined as resisting the reproduction of a single copy of a state or its functionality into two copies. But this need not be the most general setting of cloneability --- arbitrary classical states can be cloned any number of times. This makes them distinct from arbitrary quantum states, which in general cannot be copied into any larger number of copies by a quantum procedure. Therefore, it is meaningful to consider $k\to k+\ell$ uncloneability, where given $k$ copies of the state in question, no cloner is able to produce $k+\ell$ copies, for some $\ell \geq 1$. Note that for $\ell > 1$, $1\to 2$ uncloneability does not necessarily imply $k \to k+\ell$ uncloneability, or vice versa. Certainly if $k$ is exponential in the number qubits in the state, it may be possible to just do full tomography on the state and then produce as many copies as wanted. On the other hand, for some family of states, it may be possible to produce two copies given one copy, but not three copies. In the most ideal scenario, any kind of cloning attack from $k$ copies to more than $k$ would be prevented, for any polynomial $k$. 

\para{Cloneability and uncloneability for languages.} In all previous works on uncloneable proof and advice, uncloneability is a property of the specific proof or advice system and not the language. As noted earlier, the constructions in previous works are specifically for languages that also have cloneable proof or advice systems. This means that there exist languages that have both cloneable and uncloneable proof and advice systems. Depending on the definition of cloneability we use, some languages may have only cloneable proofs or advice: if we consider cloning the functionality of the proof or advice as in \cite{BKL23, BGPS24} (also called pirating in \cite{BGPS24}), then languages in BQP only have cloneable proofs, since they don't actually need proofs and the trivial functionality can always be cloned. If we consider cloning the given proof or advice state itself as in \cite{ITCS:NZ24}, no language can have only cloneable proof or advice, since we can always attach an uncloneable quantum state to an otherwise cloneable proof or advice.

This naturally raises the question: do there exist languages that only have uncloneable proof or advice states?

\subsection{Our Results}
In this work, we introduce and examine the notions of quantum proofs and advice that \emph{strictly} need to be uncloneable. Specifically, we define strictly uncloneable versions of QMA, BQP/qpoly and FEQP/qpoly, which \emph{only} admit uncloneable proofs and advice respectively. Like \cite{BKL23,BGPS24}, we consider uncloneability of the functionality of proof and advice states, since this is the strongest possible definition of uncloneability. Note that if a single copy of a state is valid as a proof for one verifier, then $k$ copies of the state constitute a valid proof
state jointly for $k$ verifiers. But some other state, for example one where the $k$ registers are entangled, can also be a valid proof jointly.
The verifiers in our definitions are naturally required to be polynomial-time; the cloning algorithm is also required to be polynomial-time for reasons we discuss in Section~\ref{sec:def}.

\para{Definitions of strictly uncloneable classes.} A priori it is not clear how to define the strictly uncloneable classes, and whether there should be any languages in them. As a first attempt, we can try to say that for strictly uncloneable QMA (suQMA), for any candidate family of proofs, no polynomial-time cloner should be able to clone it for joint use by two verifiers. However, such a definition is obviously infeasible, because if $\ket{\psi_x}$ is a valid proof state for input $x$, then so is $\ket{\psi_x}^{\otimes 2}$ --- an honest verifier can just ignore the second copy. A cloner does not even have to do anything to make this family of proofs jointly usable by two verifiers. However, if a candidate proof state is polynomial-sized, then it can only contain polynomially many copies of $\ket{\psi_x}$. So if the states $\ket{\psi_x}$ themselves are uncloneable, and if a candidate family of proof states $\{\ket{\phi_x}\}_x$ satisfies $\ket{\phi_x} = \ket{\psi_x}^{\otimes (k-1)}$, then a cloner cannot produce $\ket{\psi_x}^{\otimes k}$ from these. Inspired by this idea, we will define suQMA as a subclass of QMA such that for any polynomial-sized family of candidate proof states $\{\ket{\phi_x}\}_x$, there exists some polynomial $k$ such that a polynomial-time cloner acting on $\ket{\phi_x}$ cannot produce states $\ket{\rho_x}$ that are jointly usable as proofs by $k$ separate verifiers which act on different registers of $\ket{\rho_x}$.
Note that a general family of candidate proof states need not be $k-1$ copies of some canonical proof state $\ket{\psi_x}$, and as noted earlier, any state that is usable by $k$ verifiers to jointly solve the QMA problem on $x$ need not be $\ket{\psi_x}^{\otimes k}$. So it is still not obvious whether there should be any language in suQMA even with this definition.

Our general stance in this work is that it should be enough to consider worst-case uncloneability, i.e., we are fine if a cloner fails to provide a proof or advice state that works for some input of some length, as is the approach usually taken in complexity. For BQP/qpoly we are however able to consider a somewhat stronger approach: we require that there exists some $n$ such that the cloner fails \emph{on average} to clone advice states for inputs of that length. However, this is still weaker than the notion considered in \cite{BKL23}, who required this for every $n$. For FEQP/qpoly, we require success probabilities to be 1, so worst-case vs.\ average case makes no difference. That is, in FEQP/supoly we require the cloner to produce states that are jointly usable by $k$ algorithms to solve the relational problem with \emph{zero error} (whereas for QMA and BQP/qpoly we are okay with bounded error). There obviously exist advice states which would allow $k$ algorithms to solve the problem with perfect correctness together (since a problem in strictly uncloneable FEQP/qpoly would also be in FEQP/qpoly), so essentially in this definition we are also requiring the cloning operation to be exact. For formal definitions, see Section \ref{sec:def}.


\para{Languages in the strictly uncloneable classes.} Given our definitions, if a problem is shown to exist in suQMA, it would necessarily separate QMA and QCMA (and indeed QMA and cloneableQMA), since problems in QCMA admit at least one classical proof that is arbitrarily cloneable. Thus it is a much harder question and it is unlikely that we would be able to show this without oracles. Existence of a language in suQMA with respect to a classical oracle would imply a classical oracle separation between QMA and QCMA, which is currently not known either. Thus, we study the classes suQMA and strictly uncloneable BQP/qpoly in the quantum oracle setting.


Our first result is to show the existence of a language in suQMA in the presence of a quantum oracle. The oracle problem we use is the same one used by \cite{AK07}, which is a quantum version of the OR problem. This result also has the effect of separating the class cloneableQMA introduced in \cite{ITCS:NZ24} from QMA, which was an implicit open question in that work. The relationship between the classes QCMA, cloneableQMA, suQMA and QMA are shown in Figure \ref{fig:QMA}.

Our second result is to show, again in the presence of a quantum oracle, the existence of a language in average-case strictly uncloneable BQP/qpoly. Although our notion of uncloneability is stricter, we stress that our definition of average-case is weaker than that of \cite{BKL23} (and also we use an oracle), which makes our result incomparable to theirs.

Our final result is in demonstrating a language in the class FEQP/supoly, without the aid of any oracles. We use a parallel-repeated version of the Hidden Matching problem used by \cite{ABK23} to unconditionally separate FBQP/qpoly and FBQP/poly for this. Demonstrating that the relation is in FEQP/qpoly means that for any family of polynomial-sized candidate advice states, there is some polynomial $k$ such that a cloner cannot produce states that are usable by $k$ separate polynomial-time algorithms to solve the problem \emph{with zero error}. We suspect it should also not be possible for a cloner to produce states that can be used by $k$ algorithms to solve the problem with joint bounded error (which would place the problem in a strictly uncloneable version of FBQP/qpoly), and we leave proving this as an interesting open problem.

\begin{figure}[!ht]
\centering
\begin{tikzpicture}[scale=0.8]
    \draw[thick] (0, 0) ellipse (5 and 3);
    \node[anchor=south east] at (3, 1.5) {QMA};
    
    \draw[thick] (-1, 0) ellipse (3 and 2) node[right] {cloneableQMA};
    
    \draw[thick] (-2, 0) circle (1) node {QCMA};
    
    \draw[thick] (3.5, 0) circle (1) node {suQMA};
    
\end{tikzpicture}
\caption{Cloneable and uncloneable subclasses of QMA: in the quantum oracle model, QMA $\neq$ QCMA was shown in \cite{AK07}, cloneableQMA $\neq$ QCMA was shown in \cite{ITCS:NZ24}. In this work, we show a problem is in suQMA in the quantum oracle model, which separates cloneableQMA and QMA, since suQMA is contained in QMA$\setminus$cloneableQMA.}
\label{fig:QMA}
\end{figure}
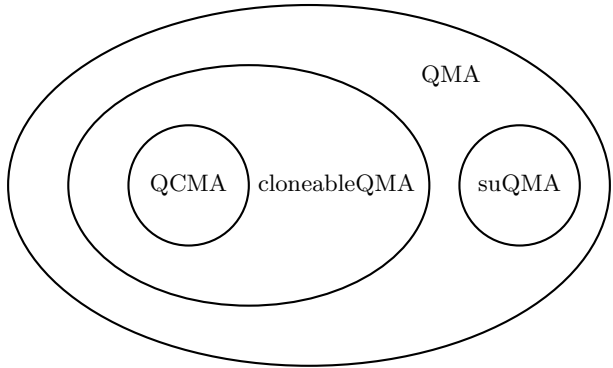

\paragraph{} Our results essentially show that the uncloneability of quantum states, which has found so many applications in cryptography, is also useful in computation --- certain languages are only efficiently computable with uncloneable quantum states as proofs and advice.

\subsection{Our techniques}
\subsubsection{Strictly uncloneable oracle QMA and BQP/qpoly}
\para{Oracle problem.} In order to show that there exists a language in suQMA with respect to a quantum oracle, we only need to show such a statement in the model of query complexity with quantum oracles --- after that we can use standard diagonalization arguments. The $\QOR_n$ problem introduced by \cite{AK07} in the query model is as follows: we are given an oracle that is either the identity matrix, or $\Id - 2\state{\psi}$ for some $n$-qubit Haar-random state $\ket{\psi}$, and we have to determine which of these is the case. \cite{AK07} showed that the state $\ket{\psi}$ is a quantum proof state for this problem. Now Haar-random states are information-theoretically uncloneable \cite{Wer98} (i.e., they are uncloneable even by computationally unbounded cloning algorithms). This sounds promising for us, because if \emph{any} possible family of proof states has to be uncloneable, the `canonical' one should definitely be uncloneable.

\para{Rigidity of proofs.} To make use of the fact that Haar-random states are uncloneable, we need to show that in some sense the state $\ket{\psi}$ is the \emph{only} possible proof state when the oracle is $\Id - 2\state{\psi}$, i.e., a `rigidity' result.\footnote{We borrow the term from the literature on non-local games.} The sense in which we show this is the following: any verifier that can distinguish $\Id - 2\state{\psi}$ and $\Id$, can also produce a copy of the state $\ket{\psi}$. It is not difficult to show this using a version of the hybrid method for query complexity \cite{BBV97}. This means that if $k$ verifiers are taking in the state provided to them by the cloner corresponding to $\Id - 2\state{\psi}$ and each of them separately distinguishing the oracle from $\Id$, then the $k$ verifiers can also jointly produce the state $\ket{\psi}^{\otimes k}$.

\para{Uncloneability of Haar-random states with oracles.} Werner \cite{Wer98} showed that Haar-random states are $(k-1)\to k$ uncloneable, i.e., given $k-1$ copies of an $n$-qubit Haar-random state, it is not possible to produce $k$ copies, for any $k=\poly(n)$. But there are two barriers to directly using this result. First, we want to show that is should not be possible to produce $\ket{\psi}^{\otimes k}$ starting from an \emph{arbitrary} polynomial-qubit state $\ket{f_\psi}$ (where $k$ depends on the polynomial size). Second, we are in the oracle model, so all verifiers will have access to the oracle $\Id - 2\state{\psi}$, which may help them clone $\ket{\psi}$. It was shown in \cite{AC13} that Haar-random states $\ket{\psi}$ are uncloneable even in the presence of the oracle $\Id - 2\state{\psi}$. But their result also only applies in the case that we are starting from states $\ket{\psi}^{\otimes (k-1)}$ (up to acting the same unitary on every starting state) and trying to produce $\ket{\psi}^{\otimes k}$, so we cannot directly use this result either.

Instead we will generalize the result of \cite{Wer98} for cloning of a discrete set of states, and then use an approach taken by \cite{AKL23} to show that uncloneable encryption schemes remain secure in the presence of an oracle that is a point function for the encrypted message. Specifically, we'll consider the problem of producing $\ket{\phi_i}^{\otimes k}$ for a $k$-design $\ket{\phi_i}$, starting from an arbitrary family of states $\ket{f_i}$ whose size is less than the log of the dimension of the $(k-1)$-fold symmetric subspace of $n$ qubits. For such a family of states, acting a cloning map on $\ket{f_i}$ and then projecting onto the states $\ket{\phi_i}^{\otimes k}$ (which measures the success probability of cloning) is effectively trying to learn the index $i$ by doing a measurement on the states $\ket{f_i}$ (up to normalization). Because the dimension of the states $\ket{f_i}$ is not too high, not all the $\ket{f_i}$-s are orthogonal, and we can bound the probability of learning $i$ from this measurement. This gives us an average-case bound on producing $\ket{\phi_i}^{\otimes k}$ from $\ket{f_i}$, which is good enough for us, as we only need a worst-case result for producing $\ket{\psi}^{\otimes k}$ starting from $\ket{f_\psi}$ for Haar-random states.

Like Werner's result, our result on producing $\ket{\psi}^{\otimes k}$ starting from arbitrary (low-dimensional) $\ket{f_\psi}$ holds without any assumptions about efficiency. If we can show that the same result holds in the presence of the oracle $\Id - 2\state{\psi}$, as long as polynomially many queries are being made to the oracle, then we should be done (recall we argued that if $k$ verifiers can solve $\QOR_n$, then they can also produce $\ket{\psi}^{\otimes k}$). Following \cite{AKL23}, we consider the overall query algorithm implemented by the cloner and the $k$ verifiers, which makes polynomially many queries, since $k$ is polynomial. Using the hybrid method again, we can argue there exists a $2^{n/2}$ dimensional subspace $S$ of the $n$-qubit Hilbert space, such that $S$ contains $\ket{\psi}$, and the query algorithm cannot distinguish the oracle $\Id - 2\state{\psi}$ from the oracle $\Id - 2S$. Therefore, we might as well give the algorithm access to $\Id - 2S$ instead of $\Id - 2\state{\psi}$. But giving the algorithm access to $\Id - 2S$ is no more powerful than telling the algorithm what $S$ is and then allowing it to make no queries, since in that case it can implement $\Id - 2S$ itself. The problem of creating $k$ copies of an $n$-qubit state that comes from a known $2^{n/2}$-dimensional subspace without making any queries then just reduces to the problem of creating $k$ copies of an unknown $n/2$-qubit state, which is impossible as we previously argued.

\para{Oracle problem with advice.} The oracle construction and proof for avBQP/supoly is very similar. To compute a random language $L: \{0,1\}^n \to \{0,1\}$, we give query access to the oracle $U_{\psi, L}$, which takes $\ket{x}\ket{\psi}$ to $(-1)^{L(x)}\ket{x}\ket{\psi}$, and acts like identity whenever the second register is orthogonal to $\ket{\psi}$, for some Haar-random $\ket{\psi}$. We can argue that any $k$ algorithms that can jointly compute $L(x)$ \emph{on average} over $x \in \{0,1\}^n$ from this oracle and some joint advice state, can also produce $\ket{\psi}^{\otimes k}$. The rest of the argument is exactly like the oracle suQMA case.

\subsubsection{Strictly uncloneable FEQP/qpoly}
\para{Connection to one-way communication complexity.}
The Hidden Matching relational problem $\HMn$, which exponentially separates classical and quantum one-way communication complexity. is as follows: the player Alice gets as input a uniformly random function $f:\{0,1\}^n \to \{0,1\}$, and player Bob gets as input uniformly random $x\in \{0,1\}^n$. Alice is supposed to send a classical or quantum message depending on her input to Bob, who needs to output $(y, f(y)\oplus f(y\oplus x))$ for any $y\in \{0,1\}^n$, where $y\oplus x$ denotes bitwise XOR. \cite{BJK08} showed that there is a quantum one-way protocol that solves this problem with $n$-qubit communication from Alice to Bob, whereas any deterministic classical protocol would require $\Omega(2^{n/2})$ bits of communication to succeed with high probability over the uniform distribution on inputs. \cite{ABK23} then observed that this problem (for some fixed $f$ for every $n$) can be used to separate $\FBQPq$ and FBQP/poly, since if a polynomial-sized classical advice state existed for every $f$, those advice states could be used as Alice as classical messages to solve $\HMn$ in communication complexity.

\para{Weak rigidity of advice states.}
The $n$-qubit quantum one-way message (advice) states given by \cite{BJK08} that solve $\HMn$ are $\frac{1}{2^{n/2}}\sum_{x\in \{0,1\}^n}(-1)^{f(x)}\ket{x}$. For a random function $f$, these states are called {\em binary phase states} and are known to be uncloneable \cite{JLS18, BS19}. We do however have the same problem as before that these may not be the only possible message states that it is possible to solve the problem with, and we may not be given $k-1$ copies of these states to begin with. We address the first problem by proving a very weak form of a rigidity result for the possible message states. Essentially, following techniques used by Yamakawa and Zhandry \cite{YZ22} in order to get a `proof of min-entropy' result from their separation between FBPP and FBQP for a random oracle, we are able to show that in any one-way protocol for $\HMn$ which has constant correctness, the output distribution of Bob must have high min-entropy (i.e., the maximum probability of outputting any $(y,f(y)\oplus f(y\oplus x))$ must be $2^{-\Omega(n)}$). With this result, even though we are not able to not able to show that the message states must be close to binary phase states, we can argue that just before the final measurement corresponding to $x$, Bob's state must look like $\ket{\phi_{f,x}} = \sum_y\alpha_{f,x}(y)\left(\frac{\ket{y}+(-1)^{f(y)\oplus f(y\oplus x)}\ket{y\oplus x}}{\sqrt{2}}\right)\ket{\phi'_{f,x,y}}$, for some arbitrary states $\ket{\phi'_{f,x,y}}$ outside the output register, where $|\alpha_{f,x}(y)|^2$ is a high min-entropy distribution (this is for zero-error protocols, for bounded error protocols, the state would instead be close to the above state). The canonical proof states are states of this form (with the $\ket{\phi'_{f,x}}$ part being empty), but other states of this form can be quite different, and are not necessarily uncloneable. Despite this, these states do satisfy a property we will see later will be useful: the pairwise inner product of $\ket{\phi_{f,x}}$ and $\ket{\phi_{f',x}}$, for random and distinct $f,f'$, is close to $\frac{1}{2}$ for every $x$ with high probability.

The idea behind proving the min-entropy lower bound is as follows: if the output distribution had low min-entropy, then it would be possible to approximate the distribution by sending not too many copies of the quantum message state. The lexicographically smallest output string which has its estimated probability above some threshold is now some almost-deterministic function of $f$ and $x$ (it is not exactly deterministic because the procedure to approximate the distribution is not exactly correct, but we can deal with this by randomly deciding the threshold, as shown in \cite{YZ22} --- the function is then deterministic for most choices of the threshold). Moreover, this string is also a correct output for $\HMn$ with inputs $f, x$ with high probability, due to the correctness of the original protocol. Each bit of this output is then also a deterministic boolean function, and concatenating the outputs of all these boolean functions will produce a correct output for $f, x$ with high probability. But it was shown in \cite{Aar05, KP14} that if there is a quantum one-way protocol for a (total or partial) boolean function of not too much communication, then there is also a deterministic one-way protocol for this boolean function whose communication complexity is not much higher (the increase in communication is basically a factor of the size of Bob's input, which is fine in this case because Bob's input is only $n$ bits).\footnote{The equivalent step in the proof of \cite{YZ22} needs the Aaronson-Ambainis conjecture for query complexity, but we are able to proceed without any unproven conjectures in one-way communication complexity.} By running this deterministic protocol bit by bit, we can reconstruct a whole output for $\HMn$ on $f, x$ with not too much communication, which contradicts the $\Omega(2^{n/2})$ deterministic lower bound for $\HMn$.

\para{Generalizing to $k$-receiver one-way communication.}
We now need to extend this result for $k$ algorithms which are all jointly computing $\HMn$ for the same $f$ (because recall $f$ will eventually be fixed for every $n$), and different $x^1, \ldots, x^k$. This corresponds to what we call a $k$-receiver one-way communication setting, where Alice has a single input, and there are $k$ Bobs, who are separated from each other and each has his own input. Alice wants to send an overall quantum message to the $k$ Bobs (which may be entangled between the registers corresponding to each) such that each of them can produce a valid output for the relation with respect to $f$ and his own input. A somewhat similar setting was recently studied in \cite{gilboa2024data}, which was called `asymmetric direct sum' in that work. However, an important distinction between their setting and ours is that they do not require the Bobs to be separated and each acting on their own registers, which allows for more efficient protocols. We ideally would like the states that allow $k$ receivers to compute $\HMn$ to be $k$ copies of the state that allows one player to compute $\HMn$. This means the min-entropy of the joint output distribution of the $k$ receivers must be $\Omega(nk)$. Unfortunately, we are unable to show such a min-entropy lower bound for the $k$ outputs on average over $f$ and $x^1, \ldots, x^k$. But if we use an $n$-parallel repeated version of $\HMn^n$, we can show a $\Omega(kn)$ for a $2^{-\poly(n)}$ fraction of $f$, which is good enough for us. Note that this is not the optimal lower bound for $n$-parallel repeated $\HMn$, because we should expect the min-entropy for a single receiver's output in that case to be $\Omega(n^2)$, but this loss is fine for us. Moreover, while the single-receiver min-entropy bound held for bounded-error protocols, we can only prove the $k$-receiver joint min-entropy bound for protocols than are jointly zero-error, which is part of the reason why are results are for FEQP/supoly rather than FBQP/supoly.

\para{Small inner products require high dimension.}
We previously argued that for $f \neq f'$, the (modulus of the) inner product of the states $\ket{\phi_{f,x}}$ and $\ket{\phi_{f',x}}$ corresponding to a single receiver is close to $\frac{1}{2}$ with high probability. Using the min-entropy bound on the output distribution of $k$ players, we can argue that the inner product between the joint states $\ket{\phi_{f,x^1\ldots x^k}}$ and $\ket{\phi_{f',x^1\ldots x^k}}$ for $k$ players corresponding to $f$ and $f'$ for a fixed $x$ is close to $2^{-\Omega(k)}$ with high probability.\footnote{Given that we are now using parallel-repeated $\HMn^n$, the inner product for a single receiver should actually be close to $2^{-n}$, but we can't say the inner product between the $k$-receiver states is close to $2^{-nk}$ because of the $n$-factor loss in min-entropy.} We now need to argue that there is some $k$ depending on $p(n)$ such that such states cannot be produced from advice states $\ket{\psi_f}$ of $p(n)$ qubits, after the action of a cloning algorithm (that does not depend on $f$) and the unitaries applied by the $k$ receivers for some fixed input. In fact, $\brakett{\psi_f}{\psi_{f'}} = \brakett{\phi_{f,x^1\ldots x^k}}{\phi_{f',x^1\ldots x^k}}$ because the same cloning map is applied to both $\ket{\psi_f}$ and $\ket{\psi_{f'}}$, and then the same unitaries corresponding to inputs $x^1, \ldots, x^k$ is applied in both cases by the $k$ verifiers. This means $|\brakett{\psi_f}{\psi_{f'}}| = 2^{-\Omega(k)}$ for many $f,f'$. Picking $k=O(p(n))$, this should not be possible because the dimension of a vector space that can support $2/\delta$ unit vectors whose pairwise inner products are most $\delta$, is at least $\Omega(1/\delta^2)$ (see e.g. \cite{dW01}).

\begin{remark}
It was observed in \cite{LLPY23} that the Yamakawa-Zhandry problem used in \cite{YZ22} to give their separation between FBPP and FBQP with respect to a random oracle, also separates one-way quantum and classical communication complexity, and hence FBQP/qpoly and FBQP/poly. All our uncloneability results should hold for Yamakawa-Zhandry problem and its advice states as well, except for the fact that this problem is not known to be in FEQP/qpoly, so it does not make sense to demand that the cloner produce states that allow $k$ players to exactly solve the problem (for all we know, such polynomial-sized states may not exist). If it is possible to generalize our proof for Hidden Matching to work for bounded error, however, we suspect a similar proof would for for the Yamakawa-Zhandry problem as well.
\end{remark}

\subsection{Organization of the paper}
In Section \ref{sec:prelim} on present some preliminaries on quantum information. In Section \ref{sec:def} we present our definitions of strictly uncloneable QMA, BQP/qpoly and FEQP/qpoly, and discuss their relationships with existing notions of cloneable and uncloneable proofs and advice. In Section \ref{sec:suQMA}, we prove our result (Theorem \ref{thm:diag-suQMA}) about the existence of a language in oracle $\suQMA$. In Section \ref{sec:BQPsuq}, we prove our result (Theorem \ref{thm:diag-BQPsuq}) about the existence of a language in oracle $\BQPsuq$. In Section \ref{sec:FEQPsuq} we present our final result (Theorem \ref{thm:diag-FEQPsuq}) about the unconditional existence of a language in $\FEQPsuq$.

\section{Preliminaries}\label{sec:prelim}

The $\ell_1$ distance between two quantum states $\rho$ and $\sigma$ is given by
\[ \Vert\rho-\sigma\Vert_1 = \Tr\sqrt{(\rho-\sigma)^\dagger(\rho-\sigma)} = \Tr|\rho-\sigma|.\]
The fidelity between two quantum states is given by
\[ \F(\rho,\sigma) = \Vert\sqrt{\rho}\sqrt{\sigma}\Vert_1.\]
$\ell_1$ distance and fidelity are related in the following way.
\begin{lemma}[Fuchs-van de Graaf inequality]\label{fvdg}
For any pair of quantum states $\rho$ and $\sigma$,
\[ 2(1-\F(\rho,\sigma)) \leq \Vert\rho-\sigma\Vert_1\leq 2\sqrt{1-\F(\rho,\sigma)^2}.\]
For two pure states $\ket{\psi}$ and $\ket{\phi}$, we have
\[ \Vert\state{\psi} - \state{\phi}\Vert_1 = \sqrt{1 - \F\left(\state{\psi},\state{\phi}\right)^2} = \sqrt{1-|\brakett{\psi}{\phi}|^2}.\]
\end{lemma}
The quantity $\sqrt{1-F(\rho,\sigma)^2}$ is sometimes called the purified distance $P(\rho,\sigma)$. It is a metric that satisfies standard properties like triangle inequality, invariance under isometries etc.

\begin{lemma}[Uhlmann's theorem]\label{uhlmann}
Suppose $\rho$ and $\sigma$ are mixed states on register $X$ which are purified to $\ket{\rho}$ and $\ket{\sigma}$ on registers $XY$, then it holds that
\[ \F(\rho, \sigma) = \max_U|\bra{\rho}\Id_X\otimes U\ket{\sigma}|\]
where the maximization is over unitaries acting only on register $Y$.
\end{lemma}

\begin{lemma}[Gentle measurement lemma]
If measurements operators $M_1, \ldots, M_k$ succeed on a state $\ket{\psi}$ with probability $1-\delta$ each, then the probability that $M_1, \ldots, M_k$ all succeed when the measurements are applied in sequence is $1-k\sqrt{\delta}$.
\end{lemma}

\subsection{Entropic quantities}
The von Neumann entropy of a quantum state $\rho$ on a register $Z$ is given by
\[ \H(\rho) = -\Tr(\rho\log \rho).\]
We shall also denote this by $\H(Z)_\rho$. For a state $\rho_{YZ}$ on registers $YZ$, the entropy of $Y$ conditioned on $Z$ is given by
\[ \H(Y|Z)_\rho = \H(YZ)_\rho - \H(Z)_\rho\]
where $\H(Z)_\rho$ is calculated w.r.t. the reduced state $\rho_Z$. 
\begin{lemma}[Data processing inequality]
If $\sigma_{YZ'} = \Id\otimes \Lambda(\rho_{YZ})$, where $\Lambda$ is some CPTP map taking the register $Z$ to $Z'$, then $\H(Y|Z)_\rho \leq \H(Y|Z')_\sigma)$. In particular, because tracing out a register is a CPTP map, we have $\H(Y|Z)_\rho \leq \H(Y)_\rho$.
\end{lemma}

The relative entropy between two states $\rho$ and $\sigma$ of the same dimensions is given by
\[ \D(\rho\Vert \sigma) = \Tr(\rho\log\rho) - \Tr(\rho\log\sigma).\]
The relative min-entropy between $\rho$ and $\sigma$ is defined as
\[ \D_\infty(\rho\Vert\sigma) = \min\{\lambda : \rho \leq 2^\lambda\sigma\}.\]
It is easy to see that for all $\rho$ and $\sigma$,
$ 0 \leq \D(\rho\Vert\sigma) \leq \D_\infty(\rho\Vert\sigma).$ 
\begin{lemma}\label{fc:event-prob}
If $\sigma = \eps\rho + (1-\eps)\rho'$, then $\D_\infty(\rho\Vert \sigma) \leq \log(1/\eps)$.
\end{lemma}
\begin{lemma}\label{fc:Sinfty-tri}
For any three quantum states $\rho, \sigma, \phi$ such that $\supp(\rho) \subseteq \supp(\phi) \subseteq \supp(\sigma)$,
\[ \D_\infty(\rho\Vert\sigma) \leq \D_\infty(\rho\Vert\phi) + \D_\infty(\phi\Vert\sigma). \]
\end{lemma}

The mutual information between $Y$ and $Z$ with respect to a state $\rho$ on $YZ$ can be defined in the following equivalent ways:
\[  \I(Y:Z)_\rho = \D(\rho_{YZ}\Vert\rho_Y\otimes\rho_Z) = \H(Y)_\rho - \H(Y|Z)_\rho = \H(Z)_\rho - \H(Z|Y)_\rho.\]
The conditional mutual information between $Y$ and $Z$ conditioned on $X$ is defined as
\[ \I(Y:Z|X)_\rho = \H(Y|X)_\rho - \H(Y|XZ)_\rho = \H(Z|X)_\rho - \H(Z|XY)_\rho.\]
Mutual information can be seen to satisfy the chain rule
\[  \I(XY:Z)_\rho =  \I(X:Z)_\rho +  \I(Y:Z|X)_\rho.\]
Moreover, for any three states $\rho, \sigma, \phi$ we have,
\[ \D(\rho_{XY}\Vert\sigma_X\otimes\phi_Y) \geq \D(\rho_{XY}\Vert\sigma_X\otimes\rho_Y) \geq \I(X:Y)_\rho.\]

\section{Complexity class definitions}\label{sec:def}
\subsection{Strictly uncloneable QMA}
First we present the standard definition of QMA.
\begin{definition}[QMA]\label{def:QMA}
A language $L \subseteq \{0,1\}^*$ is in the complexity class $\QMA$ if there exists a polynomial-time quantum algorithm $\clQ$ such that the following conditions hold:
\begin{enumerate}
\item Completeness: For every $x \in L\cap \{0,1\}^n$, there exists a $p(n)$-qubit quantum proof state (for some polynomial $p(n)$) $\ket{\psi_x}$ such that $\Pr[\clQ(x, \ket{\psi_x})=1] \geq \frac{1}{2} + \frac{1}{\poly(n)}$.
\item Soundness: For every $x\notin L$ and every $p(n)$-qubit quantum state $\ket{\phi_x}$, $\Pr[\clQ(x, \ket{\phi_x})=0] \geq \frac{1}{2} + \frac{1}{\poly(n)}$.
\end{enumerate}
We say the pair $(\clQ, \{\ket{\psi_x}\}_{x\in L})$ satisfies completeness and soundness for $L$ if the above two conditions hold.
\end{definition}
It is also possible to define QMA with the probability in the completeness condition being $a$ and the probability in the soundness condition being $b$, with $a-b \geq \frac{1}{\poly(n)}$, but we shall stick to the $\frac{1}{2} + \frac{1}{\poly(n)}$ definition for parity with BQP/qpoly later. The quantity $a-b$ is called the soundness gap of the proof system and algorithm.

To define uncloneability for QMA, we will need to define joint completeness and soundness for $k$ verifiers who act on different registers of a shared proof state.
\begin{definition}[Joint completeness and soundness]\label{def:joint-succ}
For $k$ quantum algorithms $\clB_1, \ldots, \clB_k$ which each get $x$ as a classical input, and registers $B_1, \ldots, B_k$ respectively of a quantum state $\ket{\psi_x}$ (which is potentially entangled across the $k$ registers) as their quantum proofs, we use $\Pr[\clB_1\otimes\ldots \clB_k(\ket{x}^{\otimes k}\ket{\psi_x}) = (o_1, \ldots, o_k)]$ to denote the probability of $\clB_1, \ldots, \clB_k$ respectively outputting $o_1, \ldots, o_k$. Let $\{\ket{\psi_x}\}_{x\in L}$ be a set of candidate proof states for $L$ on registers $B_1, \ldots, B_k$.
We say $(\clB, \ldots, \clB_k, \{\ket{\psi_x}\}_{x\in L})$ jointly satisfy completeness and soundness for $L$ if for some polynomial $q(n)$:
\begin{enumerate}
\item For every $x \in L$,
\[ \Pr[\clB_1\otimes\ldots\otimes\clB_k(\ket{x}^{\otimes k}\ket{\psi_x}) = 1^k] \geq \frac{1}{2} + \frac{1}{q(n)};\]
\item For every $x\notin L$ and for every state $\ket{\phi_x}$ on registers $B_1\ldots B_k$,
\[ \Pr[\clB_1\otimes\ldots\otimes\clB_k(\ket{x}^{\otimes k}\ket{\psi_x}) = 0^k] \geq \frac{1}{2} + \frac{1}{q(n)}.\]
\end{enumerate}
$1/q(n)$ is the joint soundness gap for $(\clB_1, \ldots, \clB_k, \{\ket{\psi_x}\})$.\footnote{In principle we could consider the joint soundness gap to be $1/\poly(k,n)$, but we will always consider $k$ to be polynomial in $n$, so this will not matter.}
\end{definition}
\begin{remark}
We do not allow the probability of $k$ players jointly outputting the correct answer to go down exponentially in $k$. This is because we consider the ideal case to be the algorithms $\clB_1, \ldots, \clB_k$ each receiving valid proof states (though these may be entangled with the other proofs in some complicated way). In this case, each algorithm can amplify its success probability to $1-\negl(n)$ with the same proof by the result of Mariott and Watrous \cite{MW05}. Therefore, the joint success probability of the $k$ algorithms will be $1-k\cdot\negl(n) = 1-\negl(n)$ in the ideal case (for polynomial $k$), just from the union bound.
\end{remark}


Before defining strictly uncloneable QMA, we present a weaker definition of uncloneable proof systems.
\begin{definition}[Uncloneable proof system]\label{def:u-QMA}
We say a family of polynomial-sized states $\{\ket{\psi_x}\}_{x\in L}$ is an uncloneable quantum proof system for a language $L$ if there exists a polynomial-time algorithm $\clQ$ such that $(\clQ, \{\ket{\psi_x}\}_{x\in L})$ satisfies completeness and soundness for $L$, and for any  tuple of polynomial-time quantum algorithms $(\clA, \clB, \clC)$ (where $\clA$ is a quantum algorithm that takes in $x$ and $\ket{\psi_x}$ as input and outputs a state on registers $BC$; $\clB$ is a polynomial-time quantum algorithm that takes in $x$ and the quantum state in register $B$ as input, and outputs a bit; $\clC$ is analogous to $\clB$), $(\clB, \clC, \{\clA(\ket{x}\ket{\psi_x})\}_{x\in L})$ do not jointly satisfy completeness and soundness for $L$.
\end{definition}
The above definition is similar to the definition of an anti-piracy proof system in \cite{BGPS24} except for an issue of worst-case vs.\ average-case requirements as we shall shortly explain. Following \cite{BGPS24}, we have given the cloning algorithm $\clA$ access to the input $x$. Because of this, $\clA$ necessarily has to be polynomial-time, as an unbounded cloning algorithm can always just produce the state $\ket{\psi_x}$ after looking at $x$.

It is also possible to consider a definition where $\clA$ does not get access to $x$, and indeed such a definition is considered for cloneableQMA in \cite{ITCS:NZ24}. This would make the notion of uncloneability weaker, and the notion of cloneability stronger.

In \cite{BGPS24}, the definition of an anti-piracy proof system considers a polynomial-time verifier admissible for a language if there exists a family of states with which it satisfies completeness and soundness (with the completeness parameter being $1-\negl(n)$, and the soundness parameter $\negl(n)$). They call a proof system $\{\ket{\psi_x}\}_{x\in L}$ uncloneable with respect to some family of distributions $D_n$ on inputs in $\{0,1\}^n$, if for all polynomial-time cloning algorithms $\clA$ and admissible verifiers $\clB$ and $\clC$, for all $n$,
\[ \Pr_{x\sim D}\left[(\clB\otimes\clC)\left(\ket{x}\ket{x}(\clA(\ket{x}\ket{\psi_x})\right)\right] \leq \negl(n).\]
We do not explicitly require $\clB$ and $\clC$ to be admissible in Definition \ref{def:u-QMA}, but if we did, then the condition that $(\clB, \clC, \{\clA(\ket{x}\ket{\psi_x})\}_{x\in L})$ are not jointly complete and sound would imply that there exists some $x\in L$ such that the probability of $\clB$ and $\clC$ jointly outputting $(1,1)$ on that $x$ is too small. This is essentially a worst-case version of the requirement in \cite{BGPS24}.

As was observed in \cite{BGPS24}, uncloneability is property of a proof system, not a language.\footnote{Although it is possible to define a class of languages that have \emph{some} uncloneable QMA proof system: this would be in similar spirit to the class $\BQPuq$ defined in \cite{BKL23}, which we describe later. But we stress that this class can overlap with cloneableQMA.} Indeed, \cite{BGPS24} present a candidate for an uncloneable proof system for NP, which definitely has cloneable proofs. On the other hand, strict uncloneability which we present next \emph{is} a property of a language and not a proof system: specifically, it is the property that the language only has uncloneable proofs.

\begin{definition}[Strictly uncloneable QMA]\label{def:su-QMA}
A language $L\subseteq \{0,1\}^*$ is in $\suQMA$ if it is in $\QMA$, and additionally, for any set of polynomial-sized states $\{\ket{\phi_x}\}_{x\in L}$, there exists a polynomial $k(n)$ such that for any $(k(n)+1)$-tuple of polynomial-time quantum algorithms $(\clA, \clB_1, \ldots, \clB_k)$ (where $\clA$ takes in $x$ and $\ket{\phi_x}$ as input and outputs a state on registers $B_1\ldots B_k$, and each $\clB_i$ takes in $x$ and the quantum state in register $B_i$ as input, and outputs a bit), $(\clB_1, \ldots, \clB_k, \{\clA(\ket{x}\ket{\phi_x})\}_{x\in L})$ do not jointly satisfy completeness and soundness for $L$.
\end{definition}

It is necessary that the definition of strictly uncloneable QMA says that there exists some $k$ for which $k$ verifiers cannot compute the language for any given proof. If we set $k=2$ like in Definition \ref{def:u-QMA}, then no language would be in $\suQMA$. This is because if $\ket{\psi_x}$ is a valid $\QMA$ proof, then so is $\ket{\psi_x}^{\otimes 2}$ (the single verifier that receives this proof just ignores the second copy).

Now we'll present the definition of cloneableQMA due to \cite{ITCS:NZ24}, and argue that any language that is in $\suQMA$ must necessarily not be in cloneableQMA.

\begin{definition}[Cloneable QMA, \cite{ITCS:NZ24}]
A language $L\subseteq \{0,1\}^*$ is in $\cQMA$ if there exists a polynomial-time algorithm $\clQ$ and a set of $p(n)$-qubit states $\{\ket{\psi_x}\}_{x\in L}$ such that $(\clQ, \{\ket{\psi_x}\}_{x\in L})$ satisfies completeness and soundness for $L$, and moreover, there exists a polynomial-time quantum algorithm $\clA$ such that for every $x\in L$, $\bra{\psi_x}^{\otimes 2}\clA(\ket{1^{|x|}}\ket{\psi_x})\ket{\psi_x}^{\otimes 2} \geq 1 - \negl(n)$.
\end{definition}
In the above definition, $\clA$ is not given access to the input $x$ as stated earlier, but it is given access to the input length. This is the strongest possible definition of cloneability, which also includes cloning with access to the input as a subcase. This is because the definition only requires the existence of a cloneable proof system: if there is a proof system that a cloner can clone with access to the input, then there is a proof system that simply appends each input to the previous proof, which a cloner can now clone without access to the input. Moreover, the notion of cloneability in the above definition is stronger than we require of the cloners in the strictly uncloneable classes --- the cloner here is required to product two copies of the given state, rather than merely produce some bipartite state that can be used by two verifiers.



\begin{definition}[Oracle QMA, suQMA, cloneableQMA]
For any (classical or quantum) oracle $\clO$, the respective oracle versions $\QMA^\clO, \suQMA^\clO, \cQMA^\clO$ of the classes $\QMA, \suQMA$, $\cQMA$ are the same as the standard versions, except that all the polynomial time quantum algorithms involved (so the algorithms $\clQ, \clA$ and $\clB_1, \ldots, \clB_k$) have access to the oracle $\clO$, and the proof states also depend on the oracle.
\end{definition}

\begin{lemma}\label{lem:c-suQMA}
For any oracle $\clO$, $\suQMA^\clO \subseteq \QMA^\clO\setminus\cQMA^\clO$.
\end{lemma}
\begin{proof}
By definition, $\suQMA^\clO \subseteq \QMA^\clO$. So we just need to show that if $L \in \suQMA^\clO$, then it is not in $\cQMA^\clO$. Consider a language $L$ in $\suQMA^\clO$. If it were in $\cQMA^\clO$, then there would be a family of states $\ket{\psi_x}_{x\in L}$ and oracle algorithms $\clQ$ and $\clA$ such that $(\clQ^\clO, \ket{\psi_x}_{x\in L})$ satisfy completeness and soundness for $L$, and for every $x \in L$, $\bra{\psi_x}^{\otimes 2}\clA^\clO(\ket{\psi_x})\ket{\psi_x}^{\otimes 2} \geq 1 - \negl(n)$ (we are omitting the argument $1^{|x|}$ from $\clA$ for brevity).

Now for any polynomial $k$, consider applying $\clA^\clO$ to $\ket{\psi_x}$ $k-1$ times (if $\ket{\psi_x}$ is a state on $p(n)$ qubits, then $\clA^\clO(\ket{\psi_x})$ produces a state on $2p(n)$ qubits, but we can assume all successive applications act on the first $p(n)$ qubits of the output of the previous stage). We claim 
\[ \Vert (\clA^\clO)^{k-1}(\state{\psi_x}) - \state{\psi_x}^{\otimes k}\Vert_1 \leq \negl(n).\]
To see this, note that by the definition of fidelity and purified distance, we have,
\[ P(\state{\psi_x}^{\otimes 2}, \clA^\clO(\state{\psi_x})) \leq \negl(n).\]
Since $\clA^\clO$ is applying an isometry, we have,
\begin{align*}
& P\left(\state{\psi_x}^{\otimes 3}, (\clA^\clO)^2(\state{\psi_x})\right) \\
\leq & P\left(\state{\psi_x}^{\otimes 3}, \clA^\clO(\state{\psi_x})\otimes\state{\psi_x}\right) + P\left(\clA^\clO(\state{\psi_x})\otimes\state{\psi_x},(\clA^\clO)^2(\state{\psi_x})\right) \\
\leq & P\left(\state{\psi_x}^{\otimes 2}, \clA^\clO(\state{\psi_x})\right) + P\left(\state{\psi_x}^{\otimes 2},\clA^\clO(\state{\psi_x})\right) \\
\leq & \negl(n) + \negl(n).
\end{align*}
Continuing this way, we can show $P\left(\state{\psi_x}^{\otimes k}, (\clA^\clO)^{k-1}(\state{\psi_x})\right) \leq k\cdot \negl(n)$, and then by the Fuchs-van de Graaf inequality, we can upper bound the $\ell_1$ distance by $\negl(n)$ for any polynomial $k$. Therefore, for any polynomial $k$, there exists a polynomial-time oracle algorithm that can make states negligibly close to $\ket{\psi_x}^{\otimes k}$. Since $(\clQ, \{\ket{\psi_x}\}_{x\in L})$ satisfy completeness and soundness for $L$, $(\clQ, \ldots, \clQ, \{\ket{\psi_x}^{\otimes k}\}_{x\in L})$ satisfy joint completeness and soundness. Moreover, since $(\clA^\clO)^{k-1}(\ket{\psi_x})$ is negligibly close to $\ket{\psi_x}^{\otimes k}$, $(\clQ, \ldots, \clQ, \{(\clA^\clO)^{k-1}\ket{\psi_x}\}_{x\in L})$ satisfy joint completeness and soundness for $L$. This means that for any polynomial $k$, there exists a cloning algorithm and $k$ efficient verifying algorithms such that the state produced by the cloning algorithm is jointly complete and sound for $L$ along with the $k$ verifying algorithms, which contradicts the definition of $\suQMA$.
\end{proof}
The above containment may be strict. This is because a language not being in $\cQMA^\clO$ only means that there is no cloning algorithm can produce two copies of its canonical proof state, given one copy. However, it may still be possible for a cloning algorithm to produce a different state that is not two copies of the original state, but is jointly usable by the verifiers to decide the language.

\subsection{Strictly uncloneable BQP/qpoly}
\begin{definition}[BQP/qpoly]
A language $L \subseteq \{0,1\}^*$ is in the complexity class $\BQPq$ if there is a polynomial-time quantum algorithm $\clQ$ and a family of polynomial-sized quantum advice states $\{\ket{\psi_n}\}_{n \in \bbN}$ such that for every $n\in \bbN$, and every $x\in \{0,1\}^n$, $\Pr[\clQ(x,\ket{\psi_n}) = L(x)] \geq \frac{1}{2} + \frac{1}{p(n)}$ for some polynomial $p(n)$, where we set $L(x)=0$ if $x \notin L$, and $L(x)=1$ otherwise. We call $\frac{1}{p(n)}$ the bias of the advice system and the algorithm. We call the difference between the correctness probability and $\frac{1}{2}$ the bias of the algorithm and advice system.
\end{definition}
Unlike QMA, $\BQPq$ does not have error reduction that does not increase the proof size, which is why we consider a constant bias instead of a polynomial one. 

\begin{definition}[Joint average-case correctness with advice]
For $k$ quantum algorithms $\clB_1,$ $\ldots, \clB_k$ which get $x^1, \ldots, x^k\in \{0,1\}^n$ respectively as their classical inputs, registers $B_1, \ldots, B_k$ respectively of a state $\ket{\psi_{n}}$ as quantum advice, we use $\Pr[\clB_1\otimes\ldots\otimes\clB_k(\ket{x^1}\ldots\ket{x^k}\ket{\psi_{n}}) = o^1\ldots o^k]$ to denote the probability of $\clB_1, \ldots, \clB_k$ respectively outputting $o^1, \ldots, o^k$. We say $(\clB_1, \ldots, \clB_k \{\ket{\psi_{n}}_{n\in \bbN}\})$ satisfy joint average-case correctness for a language $L \subseteq \{0,1\}^*$ if for every $n\in\bbN$,
\[ \Pr_{x^1, \ldots, x^k\in \{0,1\}^n}\left[\clB_1\otimes\ldots\otimes\clB_k(\ket{x^1}\ldots\ket{x^k}\ket{\psi_{n}}) = (L(x^1),\ldots, L(x^k))\right] \geq \frac{1}{2} + \frac{1}{p(n)},\]
for some polynomial $p(n)$, where $x^1, \ldots, x^k$ are sampled independently and uniformly. $\frac{1}{p(n)}$ is the joint bias of $(\clB_1, \ldots, \clB_k, \{\ket{\psi_{n}}_{n\in \bbN}\}$.
\end{definition}

It is possible to define joint average-case correctness with the input distribution being identical for all the players (or correlated otherwise, but independent and identical seem to be the most natural choices). We are following \cite{BKL23} in not doing this, but the language for which we consider joint average-case uncloneability, also satisfies joint average-case uncloneability with an identical input distribution, as we shall briefly discuss later.

It is also possible to define joint worst-case correctness for BQP/qpoly, where no distribution is used. Indeed we will use a worst-case perfect correctness definition later for FEQP/qpoly (although this class has zero error, so average-case and worst-case are the same thing). But we do not explicitly state this for BQP/qpoly because we will work with the stronger average case definition.

Before presenting our definition of average-case strictly uncloneable BQP/qpoly, we present a definition of uncloneable BQP/qpoly. This has similarities with the class defined by \cite{BKL23}, but is still weaker. This is the definition we will strengthen to be strictly uncloneable, rather than the similar one from \cite{BKL23}. We call our classes $\BQPuq$ and $\BQPsuq$ to denote that, unlike suQMA, there is some average-case correctness requirement, although it is a weaker notion.

\begin{definition}[Average-case uncloneable BQP/qpoly]
A language $L \subseteq \{0,1\}^*$ is in the class $\BQPuq$ if there exists an algorithm $\clQ$ and a set of $p(n)$-qubit advice states $\ket{\psi_n}_{n\in\bbN}$ such that $(\clQ,\{\ket{\psi_n}\})$ satisfy (worst-case) correctness for $L$, and for any tuple of polynomial-time algorithms $(\clA,\clB,\clC)$ (where $\clA$ is a quantum algorithm that takes in $1^n$ and $\ket{\psi_n}$ as input, and outputs a state on registers $BC$; $\clB$ is a quantum algorithm that takes in $x$ and the register $B$ as input and output a bit; $\clC$ is analogous to $\clB$), $(\clB, \clC, \{\clA(1^n, \ket{\psi_n})\})$ do not satisfy joint average-case correctness for $L$.
\end{definition}
Note that by our definition, we are fine with $(\clA, \clB, \clC)$ triples failing to satisfy average-case correctness for \emph{some} input length $n$. In contrast, the definition in \cite{BKL23} required that the joint bias of $(\clB, \clC, \{\clA(1^n, \ket{\psi_n})\})$ be negligible (which is the same as saying joint correctness with polynomial bias does not hold) for all $n$. \cite{BKL23} also additionally considered a class where the correctness of $(\clQ, \{\ket{\psi_n}\})$ is negligible instead of $\frac{2}{3}$. Since doing error reduction for $\BQPq$ involves giving multiple copies of the advice, which does not necessarily preserve uncloneability, the different correctness parameters could lead to different classes. (Our constructions will all be zero-error, so they also satisfy the stronger definition of \cite{BKL23}, but we do not explicitly require this in the class definition.) Note however, that both us and \cite{BKL23} are still requiring that the correctness of $\clQ$ to be worst-case, while $\clB, \clC$ are only required to have different types of average-case correctness, which potentially makes things easier for the cloner, just like in \cite{BKL23}.

\begin{definition}[Average-case strictly uncloneable BQP/qpoly]
A language $L \subseteq \{0,1\}^*$ is in $\BQPsuq$ if it is in $\BQPq$, and additionally, for any family $\{\ket{\psi_n}\}_{n\in\bbN}$ of polynomial-sized advice states, there exists a polynomial $k(n)$ such that for any $(k(n)+1)$-tuple of polynomial-time quantum algorithms $(\clA,\clB_1,\ldots, \clB_k)$ (where $\clA$ takes in $1^n$ and $\ket{\psi_n}$ as input and outputs a state on registers $B_1,\ldots, B_k$, and each $\clB_i$ takes in some $x\in\{0,1\}^n$ and the register $B_i$ as input, and outputs a bit), $(\clB_1, \ldots, \clB_k, \{\clA(1^n, \ket{\psi_n})\}_{n\in\bbN})$ do not jointly satisfy average-case correctness for $L$.
\end{definition}
\begin{remark}
Like in the definition of $\suQMA$, we do not allow the joint success probability to go down exponentially in $k$ in $\BQPsuq$. However, unlike $\suQMA$, $\BQPq$ does not have error reduction without increasing the proof size, so this choice needs to justified. Note that $\BQPq$ still does have conventional error reduction in which an algorithm is repeated several times, and a majority vote of the outcomes is taken. Multiple copies of the advice need to be provided for this, because each time the algorithm is run it destroys the quantum advice. Our view is that if a language can be solved with cloneable advice, then a cloner $\clA$ should be able to make arbitrarily many copies of the advice. Therefore, it should be able to give each of the $k$ parties suffciently many copies of the advice so that their individual errors are low enough, and by the union bound, the joint success probability of all $k$ parties is also high enough.
\end{remark}



Our definition of $\BQPsuq$ is intuitively stronger than the definition of $\BQPuq$, but an inclusion of the form $\BQPsuq \subseteq \BQPuq$ may still not hold. This is because it's possible that there is a problem such that for every family of candidate proof states, there is a cloner that can act on these to produce a joint advice state that works for two BQP algorithms, but not $k$ many algorithms, for some bigger polynomial $k$. Since the definition of $\BQPuq$ requires that there exist advice states that cannot be cloned for use by two BQP algorithms specifically, this problem would not be in $\BQPuq$. 

\begin{definition}[Oracle BQP/qpoly, avBQP/upoly, avBQP/supoly]
For any (classical or quantum) oracle $\clO$, the respective oracle versions $\BQPq^\clO$, $\BQPuq^\clO$, $\BQPsuq^\clO$ of $\BQPq$, $\BQPuq$, $\BQPsuq$ are the same as the standard versions, except that all the polynomial time quantum algorithms involved have access to the oracle $\clO$, and the advice states also depend on the oracle.
\end{definition}

\subsection{Strictly uncloneable FEQP/qpoly}
\begin{definition}[FEQP/qpoly]
A polynomially-bounded relation $R \subseteq \{0,1\}^*\times\{0,1\}^*$ is in the complexity class $\FEQPq$ if there exists a polynomial-time quantum algorithm $\clQ$, and a family of quantum advice states $\{\ket{\psi_{n}}\}_{n\in \bbN}$ such that for all $x$ for which there exists a $y$ such that $(x,y) \in R$,
\[ \Pr[(x,\clQ(x,\ket{\psi_{n}})) \in R] = 1.\]
\end{definition}
The bounded error class $\FBQPq$ defined in \cite{ABK23} is similar to the above definition, except the probability of outputting a correct answer is required only to be $1-\eps$, and in this case the running time of the quantum algorithm as well as the size of the quantum advice are allowed to be polynomial in $1/\eps$.
\begin{definition}[Joint perfect correctness with advice]
For $k$ quantum algorithms $\clB_1, \ldots, \clB_k$ which get $x^1, \ldots, x^k$ respectively as their classical inputs, and get registers $B_1, \ldots, B_k$ of a quantum state $\ket{\psi_{n}}$ as their quantum advice, with some abuse of notation, we use $(x^1, \ldots, x^k, \clB_1\otimes\ldots\otimes\clB_k(\ket{x^1}\ldots\ket{x^k}\ket{\psi_{n}})) \in R$ to denote the event that the outputs of $\clB_1, \ldots, \clB_k$ are valid outputs for $x^1, \ldots, x^k$ respectively in $R$. We say $\{\clB_1, \ldots, \clB_k, \ket{\psi_{n}}\}$ satisfy joint perfect correctness for $R$ if for every $n \in \bbN$, and every $x^1, \ldots, x^k \in \{0,1\}^n$,
\[ \Pr[(x^1, \ldots, x^k, \clB_1\otimes\ldots\otimes\clB_k(\ket{x^1}\ldots\ket{x^k}\ket{\psi_{n}})) \in R] = 1. \]
\end{definition}
We remark that if we were trying to define a strictly uncloneable version of $\FBQPq$ instead of $\FEQPq$, it is not necessarily clear what the corresponding definition of joint correctness for $k$ players should be. If $k$ players all have independent copies of an advice state corresponding to $\eps$ error, their joint correctness probability is $(1-\eps)^k$; since there is no error amplification procedure for $\FBQPq$, it is not necessarily possible for the $k$ parties to increase their joint correctness probability. On the other hand, it seems quite difficult to show that any language is in strictly uncloneable $\FBQPq$ when the joint correctness required is so small --- indeed we don't know how to do it even if the joint correctness required is constant.

\begin{definition}[Strictly uncloneable FEQP/qpoly] A relation $R\subseteq \{0,1\}^*\times\{0,1\}^*$ is in $\FEQPsuq$ if it is in $\FEQPq$, and additionally, for any family $\{\ket{\psi_{n}}\}_{n\in\bbN}$ of polynomial-sized advice states, there exists a polynomial $k(n)$ such that for any $(k(n)+1)$-tuple of polynomial-time quantum algorithms $(\clA,\clB_1,\ldots, \clB_k)$ (where $\clA$ takes in $1^n$ and $\ket{\psi_n}$ as input, and outputs a state on registers $B_1,\ldots, B_k$, and each $\clB_i$ takes in some $x^i\in\{0,1\}^n$ and the register $B_i$ as input, and outputs a string), $(\clB_1, \ldots, \clB_k, \{\clA(1^n, \ket{\psi_n})\}_{n\in\bbN})$ do not satisfy joint perfect correctness for $R$.
\end{definition}

\section{Strictly uncloneable oracle QMA}\label{sec:suQMA}

\subsection{Strictly uncloneable QMA in query complexity}
\para{Oracle QMA problem $\QOR_n$:} Given access to an oracle that is either (yes case) $U_\psi = \Id - 2\state{\psi}$ for some Haar-random $n$-qubit state $\ket{\psi}$, or (no case) $\Id$, decide which of these is the case.

\begin{lemma}[\cite{AK07}]\label{lem:oracleQMA}
The oracle problem $\QOR_n$ is in (zero-error) QMA, with the proof being $\ket{\psi}$ when the oracle is $U_\psi$.\footnote{Technically the oracle needs to be a controlled version of the oracle described here, so that the verifier can convert this phase oracle to a bit oracle. But as was observed in \cite{AK07}, changing the oracle to a controlled oracle, or outright giving a bit oracle makes very little difference in any of the reasoning, so we shall stick to this simplified description.}
\end{lemma}
We will use the following result in order to prove our first lemma about the problem $\QOR_n$.

\begin{theorem}[Fixed point quantum search, \cite{TGP06, CRR05}]\label{thm:search}
Suppose we are given oracle access to $U_\psi = \Id - 2\Pi$ for some projector $\Pi$, an initial state $\ket{\Phi^i}$ that satisfies $\bra{\Phi^i}\Pi\ket{\Phi^i} \geq \eps$, and also access to the unitary $U_{\Phi^i} = \Id - \state{\Phi^i}$. Then there is a quantum algorithm that prepares a state $\ket{\Phi^f}$ satisfying $\bra{\Phi^f}\Pi\ket{\Phi^f} \geq 1-\delta$ using $O\left(\frac{1}{\eps}\log(1/\delta)\right)$ total oracle calls to $U_\psi$ and $U_{\Phi^i}$.
\end{theorem}

We first prove that if there is an efficient algorithm can use an arbitrary witness state to distinguish $U_\psi$ from $\Id$ given a witness state, then there is also an efficient algorithm that prepares the final state $\ket{\psi}$ starting with the same witness state and querying $U_\psi$.
\begin{lemma}\label{lem:1wit}
Suppose for problem $\QOR_n$, the state $\ket{f_\psi}$ is given as a QMA proof when the oracle is $U_\psi$, and a $p(n)$-query quantum algorithm $\clQ$ can solve the problem with this set of proofs, with $1/s(n)$ soundness gap. Then there is a quantum algorithm $\clQ'
$ that receives $\ket{f_\psi}$, makes $O\left(p(n)^3s(n)^2\log(1/\delta)\right)$ queries to $U_\psi$, and prepares a state $g_\psi$ such that $\F(\state{\psi}, g_\psi) \geq 1 -\delta$.
\end{lemma}
\begin{proof}
The algorithm $\clQ$ puts the state it receives as proof into a work register, initializes all its other registers including the query register as the $\ket{0}$ state, and then starts querying the oracle.
Let $\ket{\Phi^f_\psi}$ denote the actual overall final state of the query algorithm when it gets $\ket{f_\psi}$ for $U_\psi$ and makes its queries, and let  $\ket{\Phi^f_\Id}$ be the final state if it ran the query algorithm starting with $\ket{f_\psi}$ on the identity oracle. Since the algorithm has to distinguish $\Id$ and $U_\psi$ given $\ket{f_\psi}$ with $1/s(n)$ probability, we must have
\begin{equation}\label{eq:final-dist}
\Vert \ket{\Phi^f_\psi} - \ket{\Phi^f_\Id}\Vert_2 = \Omega(1/s(n)).
\end{equation}
On the other hand, using the hybrid method similar to \cite{AK07}, we can say
\begin{equation}\label{eq:hybrid}
\Vert \ket{\Phi^f_\psi} - \ket{\Phi^f_\Id}\Vert_2^2  \leq 4p(n)\sum_{t=1}^{p(n)}\bra{\psi}\Phi^t_\psi\ket{\psi},
\end{equation}
where $\ket{\Phi^t_\psi}i$ is the overall state of $\clQ$ just before it makes the $t$-th query to $U_\psi$, and $\Phi^t_\psi$ is its reduced state in the query register.

To see this, we'll define some hybrid states $\ket{\Sigma^t}$ which are obtained by running $\clQ$ on $\ket{f_\psi}$ with the oracle being $U_\psi$ for the first $t$ steps, and $\Id$ for the next $p(n)-t$ steps ($\ket{\Sigma^t}$ has $\psi$ dependence, but we are suppressing that for brevity). Note that $\ket{\Sigma^0} = \ket{\Phi^f_\Id}$ and $\ket{\Sigma^{p(n)}} = \ket{\Phi^f_\psi}$. We consider the distance between the states $\ket{\Sigma^t}$ and $\ket{\Sigma^{t-1}}$ that only differ on the $t$-th query, which is applied on the state $\ket{\Phi^t_\psi}$ on both. Let $\{\ket{b_i}\}_{i=1}^{2^n}$ be an orthonormal basis for the query register of $\clQ$, with $\ket{b_1} = \ket{\psi}$. We have,
\begin{align*}
\Vert\ket{\Sigma^t} - \ket{\Sigma^{t-1}}\Vert_2^2 & = \Vert (U_\psi - \Id)\ket{\Phi^t_\psi}\Vert_2^2 \\
 & = \Vert \sum_{i=1}^{2^n}\state{b_i}(U_\psi - \Id)\ket{\Phi^t_\psi}\Vert_2^2 \\
 & = \sum_{i=1}^{2^n}\Vert \state{b_i}(U_\psi - \Id)\ket{\Phi^t_\psi}\Vert_2^2 \\
 & = \sum_{i=1}^{2^n}|\bra{b_i}(U_\psi-\Id)\Phi^t_\psi(U_\psi-\Id)\ket{b_i}| \\
 & = |(-\bra{\psi}-\bra{\psi})\Phi^t_\psi(-\ket{\psi}-\ket{\psi})| + \sum_{i=2}^{2^n}|(\bra{b_i} - \bra{b_i})\Phi^t_\psi(\ket{b_i} - \ket{b_i})| \\
 & = 4\bra{\psi}\Phi^t_\psi\ket{\psi}.
\end{align*}
By triangle inequality, we have,
\[ \Vert \ket{\Phi^f_\psi} - \ket{\Phi^f_\Id}\Vert_2 \leq \sum_{t=1}^{p(n)}\Vert\ket{\Sigma^t} - \ket{\Sigma^{t-1}}\Vert_2 \leq 2\sum_{t=1}^{p(n)}\sqrt{\bra{\psi}\Phi^t_\psi\ket{\psi}}.\]
Applying Jensen's inequality on this gives us \eqref{eq:hybrid}.

From \eqref{eq:final-dist} and \eqref{eq:hybrid} we then have for any $\ket{\psi}$,
\[ \frac{1}{p(n)}\sum_{t=1}^{p(n)}\bra{\psi}\Phi^t_\psi\ket{\psi} \geq \frac{\Vert \ket{\Phi^f_\psi} - \ket{\Phi^f_\Id}\Vert_2^2}{4p(n)^2} = \Omega\left(\frac{1}{p(n)^2s(n)^2}\right).\]
Let $\clQ''$ be the algorithm that receives the same proof as $\clQ$, then samples $t\in [p(n)]$, runs $\clQ$ for $t$ steps and outputs its query register. If $\Phi'^f_\psi$ is the output state of $\clQ''$, then by definition it satisfies for each $\ket{\psi}$,
\[ \bra{\psi}\Phi'^f_\psi\ket{\psi} = \Omega\left(\frac{1}{p(n)^2s(n)^2}\right).\]
If $\ket{\Phi'^f_\psi}$ is the overall pure final state of $\clQ''$, this means that $\bra{\Phi'^f_\psi}(\state{\psi}\otimes\Id)\ket{\Phi'^f_\psi} = \Omega(1/p(n)^2s(n)^2)$.

Now we can apply Theorem \ref{thm:search} to increase the fidelity of the final state with $\ket{\psi}$, using $\state{\psi}\otimes\Id$ as the projector. Specifically, let $\clQ'$ be the algorithm that performs Grover search with $\ket{\Phi'^f_\psi}$ (the purified form of $\Phi'^f_\psi$) as the initial state of the Grover search; $\clQ'$ can achieve $\F(\state{\psi},g_\psi) \geq 1 - \delta$ using $O\left(p(n)^2s(n)^2\log(1/\delta)\right)$ many calls to both types of reflection oracles. We can do a reflection about $\ket{\psi}\otimes\Id$ with just the given oracle $U_\psi$ (since it acts as identity on the non-query registers), and a reflection about $\ket{\Phi'^f_\psi}$ we can do by running the algorithm $\clQ''$ $O(1)$ times. Since the cost of running $\clQ''$ is $p(n)$, the overall number of queries by $\clQ'$ is $O\left(p(n)^3s(n)^2\log(1/\delta)\right)$.
\end{proof}

Now we generalize Lemma \ref{lem:1wit} to $k$ verifying algorithms.
\begin{lemma}\label{lem:kwit}
Suppose for problem $\QOR_n$, when the oracle is $U_\psi$, the state $\ket{f_\psi}$ on registers $B_1\ldots B_k$ is given as the joint QMA proof for $k$ quantum algorithms $\clB_1,\ldots, \clB_k$, where each $\clB_i$ makes at most $p(n)$ queries to the oracle. If all the algorithms can jointly solve $\QOR_n$ with $\ket{f_\psi}$\footnote{Here jointly solving $\QOR_n$ is defined analogously to Definition \ref{def:joint-succ}.} with joint soundness gap $1/s(n)$, then there are quantum algorithms $\clB_1', \ldots, \clB'_k$ which receive the state $\ket{f_\psi}$, act on registers $B_1, \ldots, B_k$ respectively, make $O\left(p(n)^3s(n)^2\log(k/\delta)\right)$ queries to $U_\psi$, and jointly prepare the state $g_\psi$ such that $\F(\state{\psi}^{\otimes k},g_\psi) \geq 1-\delta$ (where it is understood that the $i$-th copy of $\ket{\psi}$ is on the same register as the output register of $\clB'_i$).
\end{lemma}
\begin{proof}
If $\clB_1, \ldots, \clB_k$ can jointly solve $\QOR_n$, then the marginal initial and final states for each $\clB_i$ must be $\Omega(1)$ far when given oracles $\Id$ and $U_\psi$. We can apply the same argument as Lemma \ref{lem:1wit} to say that each algorithm must have on average (over the queries) $\Omega(1/p(n)^2s(n)^2)$ query-weight on $\ket{\psi}$ on oracle $U_\psi$. We can therefore define the algorithms $\clB_1'', \ldots, \clB_k''$ which (independently) sample a uniform $t$ and respectively run $\clB_1,\ldots, \clB_k$ for $t$ steps and measure the query register, and subsequently the amplitude-amplified versions $\clB'_1, \ldots, \clB'_k$ of $\clB_1'', \ldots, \clB_k''$. We use the error parameter $\delta^{4}/2k^4$ for each $\clB_i$, which means the number of queries is $O\left(p(n)^3s(n)^2\log(k/\delta)\right)$ for each $\clB'_i$. The fidelity of the marginal state corresponding to the output register of $\clB'_i$ with $\state{\psi}$ is $1-\delta^{4}/2k^4$, and so by the Fuchs-van de Graaf inquality, its distance from the state $\state{\psi}$ is at most $\delta^{2}/k^2$. So if we consider the measurement operator $M_i$ that measures the output register of $\clB'_i$ and finds the state $\ket{\psi}$, it succeeds with probability $1-\delta^{2}/k^2$. Since $M_i$ does not touch the registers corresponding to the other $\clB'_j$-s, it also succeeds on the full final state $g_\psi$ with probability at least $1-\delta^{2}/k^2$. We can assume the measurements $M_i$ are applied in sequence on $\ket{g_\psi}$ since they commute anyway. By the Gentle Measurement Lemma, the probability that all the measurements succeed on $g_\psi$ is at least $1-k\cdot \delta/k = 1-\delta$. Finally, since the measurement operator $M_1\otimes\ldots\otimes M_k$ simply projects onto the state $\state{\psi}^{\otimes k}$, we must have $\F(\state{\psi}^{\otimes k},g_\psi) \geq (1-\delta)^{1/2} \geq 1-\delta$. 

\end{proof}

We next prove the following lemma about polynomial-query algorithms not being able to distinguish the oracle $U_\psi$ from an oracle $U_S = \Id - 2S$, where $S$ is the projector on to some subspace that contains $\ket{\psi}$. This will help us essentially remove the oracle from the algorithm given by Lemma \ref{lem:kwit}, so that we can use Corollary \ref{cor:opt-clone} (which we prove after this).

\begin{lemma}\label{lem:subspace}
Let $\clQ$ be a quantum query algorithm making $p(n)$ queries. Let $\ket{\Phi^f_\psi}$ be the final state (before measurement) after running the the algorithm $\clQ$ on oracle $U_\psi = \Id - 2\state{\psi}$ from some fixed initial state (this can include e.g. a proof state). Then there exists a subspace $S$ of dimension $2^{n/2}$ containing $\ket{\psi}$, such that if $\ket{\Phi^f_S}$ is the final state after running $\clQ$ on $U_S= \Id - 2S$ starting from the same initial state, then,
\[ \left\Vert \ket{\Phi^f_\psi} - \ket{\Phi^f_S}\right\Vert_2 \leq 2p(n)\cdot 2^{-n/6}.\]
\end{lemma}

\begin{proof}
We pick some orthonormal basis of the $n$-qubit Hilbert space that includes the state $\ket{\psi}$. We now pick a random $2^{n/2}$-dimensional subspace containing $\ket{\psi}$ by picking a random $(2^{n/2}-1)$-sized subset of the other basis vectors ($\ket{\psi}$ will always be included as a basis vector). Let $\{\ket{b_1}, \ldots, \ket{b_{2^n}}\}$ be the basis we've picked, with $\ket{b_1} = \ket{\psi}$; each subspace $S$ will include a $2^{n/2}$-sized subset of $\ket{b_i}$-s. If $\Phi^t_\psi$ is the reduced state in the query register of $\clQ$ just before it makes the $t$-th query to $U_\psi$, by using the hybrid method similar to the proof of Lemma \ref{lem:1wit} we can say for a fixed $S$,
\[ \Vert\ket{\Phi^f_\psi}-\ket{\Phi^f_S}\Vert_2^2 \leq 4p(n)\sum_{t=1}^{p(n)}\sum_{\ket{b_i}\in S, i \neq 1}\bra{b_i}\Phi^t_\psi\ket{b_i}.\]

Now the probability of a basis vector $\ket{b_i}$ (for $i\neq 1$) being in a random subspace $S$ is $\dfrac{{2^n \choose 2^{n/2}-2}}{{2^n \choose 2^{n/2}-1}} = \frac{2^{n/2}-1}{2^n - 2^{n/2}+2} \leq 2^{-n/3}$. Therefore, letting $I_i$ be the indicator variable for $\ket{b_i}$ being in a random $S$, we have,
\begin{align*}
\bbE_S \Vert\ket{\Phi^f_\psi}-\ket{\Phi^f_S}\Vert_2^2 & \leq 4p(n)\bbE_S\left[\sum_{t=1}^{p(n)}\sum_{\ket{b_i}\in S, i \neq 1}\bra{b_i}\Phi^t_\psi\ket{b_i}\right] \\
 & = 4p(n)\sum_{t=1}^{p(n)}\sum_{i \neq 1}\bbE_S[I_i]\bra{b_i}\Phi^t_\psi\ket{b_i} \\
 & = 4p(n)\sum_{t=1}^{p(n)}\sum_{i \neq 1}\Pr[\ket{b_i}\in S]\bra{b_i}\Phi^t_\psi\ket{b_i} \\
 & \leq 4p(n)\cdot 2^{-n/3}\sum_{t=1}^{p(n)}\sum_{i\neq 1}\bra{b_i}\Phi^t_\psi\ket{b_i} \\
 & \leq 4p(n)\cdot 2^{-n/3}\sum_{t=1}^{p(n)} \Tr(\Phi^t_\psi) = 4p(n)^2\cdot 2^{-n/3}.
\end{align*}
Therefore, there exists a $2^{n/2}$-dimensional $S$ such that $\Vert\ket{\Phi^f_\psi}-\ket{\Phi^f_S}\Vert_2 \leq 2p(n)\cdot 2^{-n/6}$.
\end{proof}

Now we'll prove a lemma about a family of states $\{\ket{\phi_i}\}_{i=1}^N$ such that starting from any states $\{\ket{f_i}\}_{i=1}^N$ with not too many qubits, it should be hard to produce $k$ copies of $\ket{\phi_i}$ on average. The family of states $\{\ket{\phi_i}\}_{i=1}^N$ we consider will be a (complex) spherical $k$-design, which approximate $k$ copies of a Haar-random state. In particular, they satisfy
\[ \frac{1}{N}\sum_{i=1}^N\state{\phi_i}^{\otimes k} = \int d\psi \state{\psi}^{\otimes k}, \]
where $d\psi$ represents the Haar measure on $n$-qubit states. Spherical $k$-designs are known to exist for large enough $N$.

\begin{theorem}
Let $\{\ket{\phi_i}\}_{i=1}^N$ be a $k$-design for $n$-qubit states. Then for any family of states $\{\ket{f_i}\}_{i=1}^N$ of at most $ m \leq \lfloor\log(d(n,k-1))\rfloor$ qubits, and the optimal cloning map $T$, the average probability of producing $\ket{\phi_i}^{\otimes k}$ by acting $T$ on $\ket{f_i}$ satisfies
\[ \sup_T\frac{1}{N}\sum_{i=1}^N\bra{\phi_i}^{\otimes k} T(\state{f_i})\ket{\phi_i}^{\otimes k} \leq \frac{k+1}{2^n+k}.\]
\end{theorem}
\begin{proof}
First we note that the optimal cloning map $T$ will have the $k$-fold symmetric subspace as its image, since the target states $\state{\phi_i}^{\otimes k}$ are all in the symmetric subspace. We know for the Haar measure on $n$-qubit states,
\[ \int d\psi \state{\psi}^{\otimes k} = \frac{\Pi_{n,k}}{d(n,k)}, \]
where $\Pi_{n,k}$ is the projector onto the symmetric subspace, and $d(n,k)$ is its dimension, which is given by $d(n,k) = {{2^n + k -1} \choose {k-1}}$. The states $\frac{d(n,k)}{N}\state{\phi_i}^{\otimes k}$ therefore satisfy by definition,
\[ \sum_{i=1}^N \frac{d(n,k)}{N} \state{\phi_i}^{\otimes k} = \Pi_{n,k}.\]
Therefore, acting the map $T$ which takes states to $\supp(\Pi_{n,k})$, and then projecting on to the states $\frac{d(n,k)}{N} \state{\phi_i}^{\otimes k}$ is a valid POVM. In particular, we have for an optimal $T$,
\[ \frac{1}{N}\sum_{i=1}^N\bra{\phi_i}^{\otimes k} T(\state{f_i})\ket{\phi_i}^{\otimes k} = \frac{1}{d(n,k)}\sum_{i=1}^N\bra{f_i}M_i\ket{f_i},\]
where $M_i = T^\dagger(\state{\phi_i}^{\otimes k})$ is a POVM element. Note that the POVM elements $M_i$ act on $m$-qubit states, and therefore must add up to the identity for $m$-qubit states. This gives us,
\begin{align*}
\frac{1}{N}\sum_{i=1}^N\bra{\phi_i}^{\otimes k} T(\state{f_i})\ket{\phi_i}^{\otimes k} & \leq \frac{1}{d(n,k)}\sum_{i=1}^N\Tr(M_i) \\
& = \frac{1}{d(n,k)}\Tr\left(\sum_{i=1}^NM_i\right) \\
& = \frac{1}{d(n,k)}\Tr(\Id_{2^m}) \\
& \leq \frac{d(n,k-1)}{d(n,k)} = \frac{k+1}{2^n + k}. \qedhere
\end{align*}
\end{proof}

For any $2^{n/2}$-dimensional subspace $S$ of an $n$-qubit Hilbert space, there is an isometry from it to an $n/2$-qubit Hilbert space. Using a $k$-design for $n/2$-qubit states, we can therefore give a $k$-design for $S$. This gives us the following corollary for cloning of a $k$-design corresponding to a known subspace.
\begin{corollary}\label{cor:opt-clone}
Let $S$ be any $2^{n/2}$-dimensional subspace of the $n$-qubit Hilbert space, and let $\ket{f_i}$ be a family of at most $\lfloor\log(d(n/2,k))\rfloor$-qubit states corresponding to each element of a $k$-design $\{\ket{\phi_i}\}_{i=1}^N$. Then for any CPTP map $T_S$ that depends on $S$ we have,
\[ \sup_T\frac{1}{N}\sum_{i=1}^N\bra{\phi_i}^{\otimes k} T_S(\state{f_i})\ket{\phi_i}^{\otimes k} \leq \frac{k+1}{2^{n/2}+k}.\]
\end{corollary}

Using all these results, we state and prove the main theorem about query $\suQMA$. The proof is similar to the proof of augmented uncloneable security in \cite{AKL23}.
\begin{theorem}\label{thm:query-suQMA}
Suppose for the problem $\QOR_n$, $\{\ket{f'_\psi}\}_\psi$ is a family of $p(n)$-qubit quantum proofs for each oracle $U_\psi$. Then there is a $k(n)=O(p(n)/n)$ such that for any algorithm $\clA$ that takes in $\ket{f'_\psi}$ and outputs $\ket{f_\psi}$ on registers $B_1, \ldots, B_k$ after making polynomially many quantum queries, and algorithms $\clB_1, \ldots, \clB_k$ that act on $B_1, \ldots, B_k$ respectively and each query the oracle polynomially many times each, $\clB_1, \ldots, \clB_k$ cannot jointly solve $\QOR_n$ with polynomial soundness gap with the proof state $\ket{f_\psi}$.
\end{theorem}

\begin{proof}
Let $k$ be the smallest integer such that $p(n) \leq \lfloor\log(d(n/2,k-1))\rfloor$. Since $d(n/2,k-1) \leq \left(\frac{2^{n/2}+k-2}{k-1}\right)^{k-1} = O(2^{nk/2})$, we have $k(n)=O(p(n)/n)$.

If $\clB_1, \ldots, \clB_k$ can jointly solve $\QOR_n$ with the polynomial soundness gap with proof states $\ket{f_\psi}$, then by Lemma \ref{lem:kwit}, there exist algorithms $\clB'_1, \ldots, \clB'_k$ which together make polynomially many queries to $U_\psi$, such that the fidelity of their output registers with $\state{\psi}^{\otimes k}$ is at least $\frac{99}{100}$. Let us use $\clB$ to denote the overall query algorithm implemented by $\clA$ and $\clB'_1, \ldots, \clB'_k$. $\clB$ takes in $\ket{f'_\psi}$ as its initial state and makes $q(n)$ many queries for some polynomial $q$. If $\ket{\Phi^f_\psi}$ denotes the final state of $\clB$ using oracle $U_\psi$, then by Lemma \ref{lem:subspace}, there exists a $2^{n/2}$-dimensional subspace containing $\ket{\psi}$ such that
\[ \Vert\ket{\Phi^f_\psi}-\ket{\Phi^f_S}\Vert_2 \leq 2q(n)\cdot 2^{-n/6}.\]
This means $|\brakett{\Phi^f_\psi}{\Phi^f_S}| \geq 1 - q(n)^2\cdot 2^{-n/3}$. If $\Phi^f_\psi$ denotes the reduced state on the output registers of $\ket{\Phi^f_\psi}$, by Uhlmann's theorem, there exists some state $\ket{\phi}$ (this depends on $\ket{\psi}$ but we are suppressing the dependence here) on the non-output registers of $\clB$ such that $\F(\state{\psi}^{\otimes k}\otimes\state{\phi}, \state{\Phi^f_\psi}) = |(\bra{\psi}^{\otimes k}\bra{\phi})\ket{\Phi^f_\psi}| \geq \frac{99}{100}$. Therefore, $\F\left(\state{\psi}^{\otimes k}, \Phi^f_S\right) \geq |(\bra{\psi}^{\otimes k}\bra{\phi})\ket{\Phi^f_S}| \geq \frac{99}{100}\cdot(1 - q(n)^2\cdot 2^{-n/3})$.

Now we can give a CPTP map $T_S$ (depending on $S$) that takes in states $\ket{f'_\psi}$ for each $\ket{\psi} \in S$, and produces states that have high overlap with $\state{\psi}^{\otimes k}$.
$T_S$ simply applies the query algorithm $\clB$ to $\ket{f_\psi}$, using $U_S$ as the oracle. $T_S$ knows the subspace $S$, so it can implement the oracle $U_S$ by itself. Finally, all registers other than the output registers of $\clB'_1, \ldots, \clB'_k$ are traced out. It is clear that $T_S(\state{f'_\psi}) = \Phi^f_S$, and we have for every $\ket{\psi} \in S$, $\F\left(\state{\psi}^{\otimes k}, T_S(\state{f'_\psi})\right) = \bra{\psi}^{\otimes k}T_S(\state{f'_\psi})\ket{\psi}^{\otimes k} \geq \frac{99}{100}\cdot(1 - q(n)^2\cdot 2^{-n/3})$. In particular, the average success probability of the cloning map $T_S$ on states corresponding to a $k$-design for $S$ high. This contradicts Corollary \ref{cor:opt-clone}.
\end{proof}

\subsection{Diagonalization}

\begin{theorem}\label{thm:diag-suQMA}
There exists a quantum oracle $U$ and a language $L$ such that $L\in \suQMA^U$.
\end{theorem}

\begin{proof}
Consider all tuples $(p, \clA, \clB_1, \ldots, \clB_{k})$ where $p(n)$ is a polynomial growth rate of the quantum proof, $k$ is the function of $p$ given by Theorem \ref{thm:query-suQMA}, $\clA$ is an oracle BQP algorithm that takes as input a $p(n)$-qubit quantum state along with an $n$-bit classical string as input, and outputs a quantum state on registers $B_1, \ldots, B_{k}$, and $\clB_1, \ldots, \clB_{k}$ are oracle BQP algorithms that take in the registers $B_1, \ldots, B_k$ as their proof register. We can assume the algorithms $\clB_1, \ldots, \clB_k$ run in $p(n)$ time and $\clA$ runs in $kp(n)$ time. Further, we can take $p(n) = n^c + c$ for some positive integer $c$ so that $p(n)$ is efficiently computable and there are countably many such growth rates. Then the set of tuples of the form we are considering is also countable: this is because, for a fixed $p$, the set of tuples $(\clA, \clB_1, \ldots, \clB_{k})$ is countable, and a union of countably many countable sets is countable. Consider a bijection $N$ from the set of these tuples to the natural numbers $n$ that are large enough for Theorem \ref{thm:query-suQMA} to be non-trivial (there are only finitely many $n$ for which it's not). For a tuple $s = (p, \clA, \clB_1, \ldots, \clB_{k})$, let $N(s)$ be the natural number the bijection maps it to. We'll construct the language $L$ and the oracle $U$ iteratively, starting from the $s$ with the minimum $N(s)$, by a diagonalization argument. For $n$ smaller than this, $L\cap\{0,1\}^n$ and $U_n$ can be picked arbitrarily, as long as it satisfies the form we describe next.

Overall, $L$ will be a unary language and for any $n$: if $0^n \in L$, then $U_n$ will be equal to $U_{\psi_n}$ for some $n$-qubit state $\ket{\psi_n}$; if $0^n\notin L$, then $U_\psi$ will be the $n$-qubit identity operator $\Id_n$ (we once again observe that the oracle technically needs to be controlled version of these unitaries). By Lemma \ref{lem:oracleQMA}, this $L$ is in $\QMA^U$, with the proof for $0^n\in L$ being $\ket{\psi_n}$. So we need to make sure for every tuple $s=(p, \clA, \clB_1, \ldots, \clB_k)$, one of the $\clB_i$-s fails to compute the language when $p(n)$-sized states are given to $\clA$, with the oracle we fix.

Assume we are at a tuple $s=(p,\clA,\clB_1,\ldots, \clB_k)$ in the iteration such that the language $L$ has been fixed on $n$-bit inputs and so has $U_n$ for $n < N(s)$. We will fix the language on $N(s)$-bit inputs and $U_{N(s)}$ in this iteration. Fixing queries that the algorithms $\clB_1, \ldots, \clB_k$ to $U_n$ for $n\neq N(s)$, we get $k$ algorithms that only make queries to the oracle at the $N(s)$-th location. By Theorem \ref{thm:query-suQMA}, $\clB_1,\ldots, \clB_k$ are not jointly complete and sound for the proof size $p(n)$ and cloning algorithm $\clA$. Let $\clB_i$ be the algorithm that misbehaves: it should either output 1 with too high probability on oracle on $\Id_{N(s)}$ or 0 with too high probability on $U_\psi$, for some $N(s)$-qubit state $\psi$. If the misbehaviour is on $\Id_{N(s)}$, we don't put $0^{N(s)}$ in $L$, and we fix $U_{N(s)} = \Id_{N(s)}$; if the misbehaviour is on $U_\psi$, then we put $0^{N(s)}$ in $L$ and fix $U_{N(s)}=U_\psi$. This ensures that the algorithm $\clB_i$ does not satisfy either soundness or completeness on $L\cap \{0,1\}^{N(s)}$. Therefore, for cloning procedure $\clA$ that acts on $p(n)$-sized proofs, $\clB_1, \ldots, \clB_k$ cannot jointly compute $L$ relative to $U$ with their proof registers from $\clA$.

Continuing this way, for each tuple $(p,\clA,\clB_1,\ldots, \clB_k)$, we ensure that for proofs of size $p$, if we have $\clA$ as the cloning procedure and $k$ algorithms $\clB_1, \ldots, \clB_k$, there is an $n$ such that $\clB_1, \ldots, \clB_k$ are not jointly complete and sound on $L \cap \{0,1\}^n$. Therefore, for any set of proof states $\{\ket{\psi_{x,U}}\}_{x\in L}$, there exists a polynomial $k(n)$ such that for any polynomial-time cloning algorithm $\clA$ and polynomial-time algorithms $\clB_1, \ldots, \clB_k$, the algorithms $\clB_i$ are not jointly complete and sound for $L$ with the joint proof states $\{\clA(\ket{x}\ket{\psi_{x,U}})\}$. That is, $L$ is in $\suQMA^U$.
\end{proof}

The above theorem and Lemma \ref{lem:c-suQMA} have the following corollary.
\begin{corollary}
There exists a quantum oracle $U$ such that $\QMA^U \neq \cQMA^U$.
\end{corollary}

\section{Strictly uncloneable oracle BQP/qpoly}\label{sec:BQPsuq}
\para{Oracle BQP/qpoly problem:} For $L: \{0,1\}^n \to \{0,1\}$ and a Haar-random $n$-qubit state $\ket{\psi}$, let $U_{\psi,L}$  be the following oracle:
\[ U_{\psi,L}\ket{x}\ket{\phi} = \begin{cases} (-1)^{L(x)}\ket{x}\ket{\phi} & \text{ if } \ket{\phi} = \ket{\psi} \\ \ket{x}\ket{\phi} & \text{ if } \brakett{\phi}{\psi} = 0.\end{cases} \]
\begin{lemma}[\cite{AK07}]\label{lem:BQP-ub}
There is an algorithm that given query access to $U_{\psi,L}$, and $\ket{\psi}$ as advice, computes $L$ with zero error on all $x\in \{0,1\}^n$ with polynomially many queries.
\end{lemma}

First we prove that a random $L$ cannot be computed with non-negligible bias with polynomial-sized quantum advice.
\begin{lemma}\label{lem:rand-qadv}
Let $L:\{0,1\}^n\to\{0,1\}$ be a random function, and let $\ket{\psi_L}$ be a $p(n)$-qubit quantum advice state depending on $L$. Then for any unbounded-time quantum algorithm $\clQ$ that takes $\ket{\psi_L}$ as advice and does not know the function $L$ otherwise, we have,
\[ \Pr_{x,L}[\clQ(x,\ket{\psi_L}) = L(x)] \leq \frac{1}{2} +  O\left(\left(\frac{p(n)}{2^n}\right)^{1/3}\right).\]
\end{lemma}
\begin{proof}
Suppose $\clQ$ computes $L$ correctly with probability $\frac{1}{2} + \eps$, with $p(n)$-qubit advice $\ket{\psi_L}$. By Markov's inequality, the probability over all $L$ that the success probability (over $x$) of the adversary is at most $\frac{1}{2} + \frac{\eps}{2}$ is at most $\frac{\eps}{2}$. Therefore, there exists a set $S$ of functions of size at least $2^{2^n-1}\cdot\eps$ on which $\clQ$ successfully predicts $L(x)$ with probability $\frac{1}{2} + \frac{\eps}{2}$ over $x$ and internal randomness. Similarly, for the functions in $S$, there exists a set $T$ of inputs of size at least $2^{n-2}\cdot\eps$ on which $\clQ$ predicts $L(x)$ successfully with probability at least $\frac{1}{2} + \frac{\eps}{4}$. The behaviour of $\clQ$ can be modelled as follows: on getting input $x$, it performs a measurement $M^x$ on the advice state $\ket{\psi_L}$, to give its guess for $L(x)$. The success probability of $\clQ$ on $(L,x)$ from $S\times T$ can be amplified to $\frac{1}{2} + \frac{1}{3}$ by repeating the procedure $\frac{2\ln(1/3)}{\eps^2}$ times, with as many copies of $\ket{\psi_L}$, and giving the majority vote of the outcomes (note that the algorithm does not need to know that the function is in the good set $S$ to do this). Let's call the algorithm which does this $\clQ'$, and its advice state $\ket{\phi^L}$ (we're switching the $L$ from the subscript to the superscript for convenience), which is $O(p(n)/\eps^2)$ qubits.

Suppose the states $\ket{\phi^L}$ are in a register $A$, and define the CQ state
\[ \phi'_{LA} = \frac{1}{2^{2^n}}\sum_L\state{L}_L\otimes\state{\phi^L}_A.\]
This state must satisfy $\I(L:A)_{\phi'} \leq 2\log|A| = O(p(n)/\eps^2)$. Taking $\sigma_{LA}$ to be the state $\phi'_{LA}$ conditioned on $L$ being in the set $S$ we have that $\phi'_{LA} = \eps\sigma_{LA} + (1-\eps)\sigma'_{LA}$ for some $\sigma'_{LA}$. This implies $\D_\infty(\sigma_{LA}\Vert\phi'_{LA}) \leq \log(1/\eps)$. Therefore,
\begin{align}
\I(L:A)_\sigma & \leq \D_\infty(\sigma_{LA} \Vert \phi'_L\otimes\phi'_A) \nonumber \\
 & \leq \D_\infty(\sigma_{LA}\Vert\phi'_{LA}) + \D_\infty(\phi_{LA}\Vert\phi'_L\otimes\phi'_{A}) \nonumber \\
 & \leq \log(1/\eps) + \I(L:A) = \log(1/\eps) + O(p(n)/\eps^2). \label{eq:Imax}
\end{align}

On the other hand, letting $L_y$ represent the value of $L(y)$, and assuming a canonical ordering of $y\in\{0,1\}^n$, we have,
\begin{align*}
\I(L:A)_{\sigma} & = \sum_{y=1}^{2^n}\I(L_y:A|L_1\ldots L_{y-1})_{\sigma} \\
 & = \sum_{y=1}^{2^n}\left(\H(L_y|L_1\ldots L_{y-1})_{\sigma} - \H(L_y|AL_1\ldots L_{y-1})_{\sigma}\right) \\
 & \geq \sum_{y=1}^{2^n}\left(\H(L_y|L_1\ldots L_{y-1})_{\sigma} - \H(L_y|A)_{\sigma}\right) \\
 & = \H(L)_\sigma - \sum_{y\notin T}\H(L_y|A)_\sigma - \sum_{y\in T}\H(L_y|A)_\sigma.
\end{align*}
Note that $\H(L)_\sigma$ is just $\log|S| \geq 2^n - \log(1/\eps)-1$. For $y\notin T$, $\H(L_y|A) \leq 1$, since $L_y$ is just a bit. For $y\in T$, let $Z_y$ denote the outcome obtained by doing the measurement $M^y$ on $\phi^L_{A}$ (which is the same as $\sigma^L_{A}$, since only the $L$ register was conditioned on to get $\sigma$ from $\phi$). Then by the data processing inequality, we have, $\H(L_y|A)_\sigma \leq \H(L_y|Z_y)_\sigma$. On the other hand, since $Z_y=F_y$ with probability at least $\frac{5}{6}$ for $y\in T$, for these $y$ we have by Fano's inequality, $\H(L_y|Z_y)_\sigma \leq h(\frac{5}{6})$, where $h$ is the binary entropy function. This gives us,
\begin{align*}
\I(L:A)_\sigma & \geq 2^n - \log(1/\eps) - 1 - \sum_{y\notin T}1 - \sum_{y\in T}h(5/6) \\
 & \geq 2^n - \log(1/\eps) - 1 - \left(1-\frac{\eps}{4}\right)\cdot 2^n - \frac{\eps}{4}\cdot 2^n \cdot h(5/6) \\
 & = \frac{\eps}{4}\cdot 2^n(1-h(5/6)) - \log(1/\eps) - 1. 
\end{align*}
Using the upper bound on $\I(A:L)_\sigma$ from \eqref{eq:Imax}, this gives us
\[
O(p(n)/\eps^2) + \log(1/\eps) \geq \frac{\eps}{4}\cdot 2^n(1-h(5/6)) - \log(1/\eps) - 1 \quad \quad \Rightarrow \quad \quad \eps^3 \leq \frac{O(p(n))}{2^n(1-h(5/6))}.
\]
This proves the required upper bound on the average-case success probability of $\clQ$.
\end{proof}

Next, using Lemma \ref{lem:rand-qadv}, we prove the equivalent of Lemma \ref{lem:1wit} for a random language with quantum advice.
\begin{lemma}
Suppose for a random function $L:\{0,1\}^n\to\{0,1\}$, there is a quantum algorithm $\clQ$ which given oracle access to $U_{\psi,L}$ and polynomial-sized $\ket{f_{\psi,L}}$ as quantum advice, and input $x$, computes $L(x)$ correctly on average (over $x$) with $\frac{1}{s(n)}$ bias (for some polynomial $s(n)$), by making $p(n)$ queries to $U_{\psi,L}$. Then there is an algorithm $\clQ'$ which receives $\ket{f_{\psi,L}}$, makes $O(p(n)^3s(n)^2\log(1/\delta))$ queries to $U_{\psi,L}$, and prepares a state $g_{\psi,L}$ such that $\bbE_L \F(\state{\psi}, g_{\psi,L}) \geq 1-\delta$.
\end{lemma}

\begin{proof}
Let $\ket{\Phi^f_{L,\psi,x,U}}$ be the overall final state of the algorithm $\clQ$ after receiving the advice $\ket{f_{\psi,L}}$ and making queries to $U_{\psi,L}$. Let $\ket{\Phi^f_{L,\psi,\Id,x}}$ be the corresponding state if the oracle $U_{\psi,L}$ were replaced by the identity oracle instead. Note that giving access to the identity oracle is equivalent to not giving the algorithm access to any oracle at all, and just giving it quantum advice $\ket{f_{\psi,L}}$. Therefore, we are in the situation described by Lemma \ref{lem:rand-qadv}: the state $\ket{\Phi^f_{\psi,L,\Id,x}}$ can be used to compute $L(x)$ with average-case (over $x$ and $L$) probability at most $\frac{1}{2}+\negl(n)$. On the other hand, by assumption, the state $\ket{\Phi^f_{\psi,L,U,x}}$ can be used to compute $L(x)$ with average-case success (over $x$, for every $L$) probability at least $\frac{1}{2} + \frac{1}{s(n)}$. Therefore, we must have,
\[ \bbE_{L,x}\Vert\ket{\Phi^f_{\psi,L,U,x}} - \ket{\Phi^f_{\psi,L,\Id,x}}\Vert_2 = \Omega(1/s(n)).\]

On the other hand, the oracles $U_{\psi,L}$ and $\Id$ only differ when the second query register is $\ket{\psi}$. So by the same type of hybrid argument as the proof of Lemma \ref{lem:1wit}, we have,
\[ \Vert\ket{\Phi^f_{\psi,L,U,x}} - \ket{\Phi^f_{\psi,L,\Id,x}}\Vert_2 \leq 2 \sum_{t=1}^{p(n)}\sqrt{\bra{\psi}\Phi^t_{\psi,L,x}\ket{\psi}}\]
where $\Phi^t_{\psi,L,x}$ is the reduced state of the second query register of $\clQ$ just before it makes the $t$-th query to $U_{\psi,L}$. This gives us,
\[ \bbE_{L,x}\left[\frac{1}{p(n)}\sum_{t=1}^{p(n)}\bra{\psi}\Phi^t_{\psi,L,x}\ket{\psi}\right] = \Omega\left(\frac{1}{p(n)^2s(n)^2}\right).\]
Now we can construct an algorithm $\clQ''$ which when given access to oracle $U_{\psi,L}$ and advice state $\ket{f_\psi}$, samples uniformly $x\in \{0,1\}^n$ and $t\in [p(n)]$, runs $\clQ$ on input $x$ for $t$ steps, and measures its second query register. On expectation over $L$, the fidelity between the output of $\clQ'$ and $\ket{\psi}$ is $\frac{1}{p(n)^2s(n)^2}$, and this can be amplified as before to get the algorithm $\clQ'$ which achieves the required fidelity.
\end{proof}

Using the above lemma, the following lemma is proved in the exact same way as Lemma \ref{lem:kwit}.

\begin{lemma}
Suppose for a random language $L:\{0,1\}^n\to\{0,1\}$, there is a polynomial-sized family of states $\ket{f_{\psi,L}}$ on registers $B_1, \ldots, B_k$, and algorithms $\clB_1, \ldots, \clB_k$ which each have access to $U_{\psi,L}$, and $\clB_i$ takes the register $B_i$ as its quantum advice. If $\clB_1, \ldots, \clB_k$ satisfy joint average-case correctness for $L$ with bias $\frac{1}{s(n)}$ with this advice, then there are algorithms $\clB'_1,\ldots, \clB'_k$ that receive the state $\ket{f_{\psi,L}}$, act on registers $B_1, \ldots, B_k$ respetively, make $O(p(n)^3s(n)^2\log(k/\delta))$ queries to $U_{\psi,L}$, and jointly prepare a state $g_{\psi,L}$ such that $\bbE_L\F(\state{\psi}^{\otimes k},g_{\psi,L}) \geq 1 - \delta$ (where it is understood that the $i$-th copy of $\ket{\psi}$ is on the same register as the output of $\clB'_i$).
\end{lemma}

\begin{remark}
Since in our definition of joint average-case correctness, the algorithms $\clB_1, \ldots, \clB_k$ receive independent inputs $x_1, \ldots, x_k$, in the proof of the above lemma, the algorithms $\clB'_1, \ldots, \clB'_k$ can sample inputs randomly, then run $\clB_1, \ldots, \clB_k$ on their sampled inputs, and measure the query registers to get $\ket{\psi}$. If we required the distribution to be identical rather than independent for joint-average case correctness instead, they would have to pre-share randomness in order to be able to sample the same random input, and then run the algorithm on that input, to be able to prepare $\ket{\psi}^{\otimes k}$. This is the only difference in the identical input case.
\end{remark}

The following theorem is proved in the same way as Theorem \ref{thm:query-suQMA}, and in fact it is true for a random language $L:\{0,1\}^n \to \{0,1\}$, rather than just worst-case as in the statement of the theorem (we only really need the worst case statement, which is why we state it that way).
\begin{theorem}\label{thm:query-suBQP}
For a polynomial $p(n)$, there is a $k=O(p(n)/n)$ such that any algorithm $\clA$ that takes in a $p(n)$-qubit quantum advice and outputs a polynomial-sized state on registers $B_1, \ldots, B_k$ after making polynomially many queries, and algorithms $\clB_1, \ldots, \clB_k$ that take in the registers $B_1, \ldots, B_k$ as their quantum advice, and make polynomially many queries, there exists an $L: \{0,1\}^n \to \{0,1\}$ and an $n$-qubit state $\ket{\psi}$ such that $\clB_1, \ldots, \clB_k$ cannot satisfy joint average-case correctness for $L$ with oracle access to $U_{\psi,L}$ and the state provided by $\clA$ as advice.
\end{theorem}

Finally, using Theorem \ref{thm:query-suBQP}, we can make a diagonalization argument.
\begin{theorem}\label{thm:diag-BQPsuq}
There exists a quantum oracle $U$ and a language $L$ such that $L \in \BQPsuq^U$.
\end{theorem}
\begin{proof}
This proof works the same way as Theorem \ref{thm:diag-suQMA}: we consider tuples $(p,\clA, \clB_1, \ldots, \clB_k)$, where $p(n)$ is the polynomial growth rate of the quantum advice, and $k$ is the function of $p$ given by Theorem \ref{thm:query-suBQP}. For a tuple $s=(p,\clA, \clB_1, \ldots, \clB_k)$, let $N(s)$ be the natural number corresponding to it. We fix the $L\cap \{0,1\}^{N(s)}$ to be the function $L_{N(s)}: \{0,1\}^{N(s)}\to\{0,1\}$ (thought of as a language) that $\clA, \clB_1, \ldots, \clB_k$ fail on when the oracle is $U_{\psi, L_{N(s)}}$, due to Theorem \ref{thm:query-suBQP}. The oracle at $N(s)$ is fixed to be $U_{\psi, L_{N(s)}}$. By Lemma \ref{lem:BQP-ub}, this language overall is in $\BQPq$, so by the diagonalization construction, it is in $\BQPsuq$.
\end{proof}

\section{Strictly uncloneable FEQP/qpoly}\label{sec:FEQPsuq}
\subsection{One-way communication complexity}
In the one-way communication complexity setting, Alice and Bob want to jointly compute a function or relation $R \subseteq \clX\times\clY\times\clZ$, for which Alice gets an input from $\clX$ and Bob gets an input from $\clY$. Alice sends a classical or quantum message depending on her input to Bob, and Bob is required to product a valid output for the relation for both of their inputs, using Alice's message. We call the total number of qubits in Alice's message for a particular protocol the communication complexity of that protocol.

We use $\D^{A\to B}(R)$ to denote the deterministic one-way communication complexity of $R$, i.e, the minimum communication complexity of a deterministic one-way communication protocol that computes $R$. $\D^{A \to B}_\eps(R,p)$ for some probability distribution $p$ on $\clX \times \clY$ is the distributional one-way  communication complexity with error $\eps$, i.e, the minimum communication of a deterministic one-way protocol that computes $R$ with probability at least $1-\eps$ over inputs from $p$. The randomized one-way communication complexity $\R^{A\to B}_\eps(R)$ with error $\eps$ is the minimum communication of a one-way randomized protocol that computes $R$ with probability at least $1-\eps$ on every input from $\clX\times\clY$. We allow Alice and Bob to share public coins in this definition. By Yao's lemma, $\R_\eps^{A\to B}(R) = \max_p\D^{A\to B}_\eps(R,p)$, so $\D^{A\to B}_\eps(R,p)$ with respect to any $p$ provides a lower bound on $\R^{A\to B}_\eps(R)$. The quantum one-way communication complexity $\Q^{A\to B}_\eps(R)$ with error $\eps$ is the minimum quantum communication complexity of a one-way quantum protocol that computes $R$ with probability at least $1-\eps$ on every input. For quantum communication complexity, we can in general consider models where Alice and Bob share entanglement, but in all the protocols we will be talking about here, they will not share prior entanglement.

\textbf{$k$-receiver one-way communication.} We will also consider a non-standard $k$-receiver one-way communication setting, where there is one sender Alice, and multiple receivers $\clB_1, \ldots, \clB_k$ who are separated. Alice gets a single input from $\clX$, and each of $\clB_1, \ldots, \clB_k$ get independent inputs from $\clY$. They are required to produce a valid output for $R$ for Alice's input, and each of their respective inputs. Alice sends a message to each of $\clB_1, \ldots, \clB_k$, and in the quantum case, the message registers may be entangled between them. But the receivers are required to produce their outputs only acting on the part of the overall message that Alice sent to them; in the quantum case this means that the measurements performed by the $k$ receivers are in tensor product. This $k$-receiver one-way model is similar to the `asymmetric direct sum' model considered in \cite{gilboa2024data} except for one critical difference: in their model, there are multiple receivers with different inputs, but they are not separated, and can act on the same registers.

Worst-case joint error or joint correctness in the $k$-receiver model is defined similar to how we've defined worst-case joint correctness elsewhere: we require that the probability that all the players output a valid output for their respective inputs with probability at least $1-\eps$. However, in the $k$-receiver case, we will always be considering perfectly correct protocols, so the joint correctness probability will always be $1$.

\subsection{One-way communication for Hidden Matching}
\para{Hidden Matching problem:} In the Hidden Matching problem ($\HMn$), Alice has input a boolean function $f:\bits^n \rightarrow \bits$ and Bob has input $x \in \{0,1\}^n$. The valid outputs for the relation are $(y, f(y)\oplus f(y\oplus x))$, for any $y\in \{0,1\}^n$.

We will now prove a result about the min-entropy of the output distribution of a one-way quantum protocol for parallel-repeated $\HMn$. To do so, we will need a result about the Hidden Matching problem due to \cite{BJK08}, as well as a result relating quantum and deterministic one-way communication complexities of boolean functions, due to \cite{Aar05, KP14}. The proof of this result will be very similar to the proof of Theorem 8.4 in \cite{YZ22}.
\begin{theorem}[\cite{BJK08}]\label{thm:HM-lb}
For the Hidden Matching problem, $\Q^{A\to B}_0(\HMn) =n$, but for $\delta<\frac{1}{8}$, we have $\D^{A\to B}_{\delta}(\HMn, D_k\otimes U) = \Omega(2^{n/2}) - k$, where $D_k\otimes U$ is the product of a distribution with $2^n-k$ min-entropy on $f$, and the uniform distribution on $x$. In particular, for $k=\poly(n)$, the distributional communication complexity is $\Omega(2^{n/2})$.
\end{theorem}
The above theorem is a generalization of the result in \cite{BJK08}, which states a result for the uniform distribution on $f$.
The uniform distribution over $f$ is only used in one place in their proof, to make the following argument: if the probability of some event over $f$ is at least $\frac{1}{2}$, then at least half of all $f$-s must be in the event. With a high min-entropy distribution, we can instead say that at least $2^{2^n - k-1}$ $f$-s are in the event, and the rest of the argument proceeds the same way as the original proof.
\begin{theorem}[\cite{Aar05,KP14}]\label{thm:QtoD}
For any partial boolean function $f:\{0,1\}^n\times\{0,1\}^m\to\{0,1\}$, if Alice gets the $n$-bit input and Bob the $m$-bit input, then $\D^{A\to B}(f) = O(m\cdot \Q^{A\to B}(f))$.
\end{theorem}

Next, we consider the $n$-fold parallel repetition of $\HMn$, which we denote by $\HMn^n$, which involves solving $n$ independent instances of $\HMn$, with functions $f_1\ldots f_n$ and inputs $x_1\ldots x_n$. With some abuse of notation, we will also be denoting the inputs of $\HMn^n$ by $f$ and $x$, where it is to be understood that $f$ consists of $n$ functions from $\{0,1\}^n$ to $\{0,1\}$; this $f$ can also be thought of as a function from $[n]\times\{0,1\}^n\to\{0,1\}$. The following is a consequence of the direct product theorem for bounded-round communication from \cite{JPY12}.
\begin{theorem}\label{thm:HM-DP}
Let $D_{1,k}, \ldots, D_{n,k}$ be distributions over functions from $\{0,1\}^n \to \{0,1\}$ each with min-entropy at least $2^n - k$, and let $D^n_k = D_{1,k}\otimes\ldots \otimes D_{n,k}$. With $U^n$ denoting the $n$-fold tensor product of the uniform distribution on $\{0,1\}^n$, we have
\[ \D^{A\to B}_{1-2^{-\Omega(n)}}(\HMn^n, D_k^n\otimes U^n) = \Omega(n(2^{n/2} - k)).\footnote{Technically, the theorem statement in \cite{JPY12} has $n$ i.i.d. copies of the same ``hard'' distribution on Alice and Bob's inputs for the $n$ inputs to $\HMn^n$, instead of independent but distinct distributions like $D_{1,k}, \ldots, D_{n,k}$. But their proof actually proves the more general version we state here, as long as the communication lower bound for the single copy of $\HMn$ holds against all the distributions $D_{i,k}$.}\]
\end{theorem}

In the following, we use $\clP(f,x)$ to denote a (probabilistic) outcome of a one-way communication protocol $\clP$ between Alice and Bob, on inputs $f$ and $x$.

\begin{theorem}\label{thm:minent-1}
Let $\clP$ be a quantum one-way quantum communication protocol for $\HMn^n$ with error at most $\frac{1}{32}$, and polynomially many qubits of communication. Let $D^n_k=D_{1,k}\otimes\ldots\otimes D_{n,k}$ be a probability distribution on $f$ such that each $D_{i,k}$ has min-entropy $2^n - k$ for some $k=\poly(n)$, and $U^n$ be the uniform distribution on $x$. Then,
\[ \Pr_{f,x\sim D_k^n\otimes U^n}\left[\max_z\Pr[\clP(f,x)=z] > 2^{-n/9} \right] = 2^{-\Omega(n)}.\]
\end{theorem}

\begin{proof}
Let the communication complexity of $\clP$ be $p(n)$. Assume for the sake of contradiction that the  satisfies
\begin{equation}\label{eq:delta-contr}
\Pr_{f,x\sim D_k^n\otimes U^n}\left[\max_z\Pr[\clP(f,x)=z] > 2^{-n/9} \right] = 2^{-o(n)}.
\end{equation}
We'll call the set of $(f,x)$ for which $\max_z\Pr[\clP(f,x)=z] > 2^{-n/9}$ the good set. Now using the good set of $(f,x)$, we'll give a classical protocol for $\HMn^n$ that violates Theorem \ref{thm:HM-DP} by a series of reductions.

First we show that there is a one-way protocol $\clQ_{\gamma,\delta}$ that has $O\left(\frac{n^2p(n)}{\gamma^2}\log(1/\delta)\right)$ qubits of communication, and in which on inputs $f, x$, Bob outputs a list $\{P^{f,x}(z)\}_{z\in \{0,1\}^{n(n+1)}}$ such that for all $f, x$,
\[ \Pr\left[\forall z \in \bits^{n(n+1)}\left|P^{f,x}(z) - \Pr[\clP(f,x) = z]\right| \leq \gamma\right] \geq 1-\delta. \]
Like in the proof of Lemma 8.5 of \cite{YZ22}, this is done by repeating the protocol $\clP$ $N$ times, where $N= \Theta\left(\frac{1}{\gamma^2}\log(2^{n(n+1)}/\delta)\right)$, and outputting $K_z/N$ as the estimate $P^{f,x}(z)$, where $K_z$ is the number of times the output $z$ occurs. By the Chernoff bound, for each $P^{f,x}(z)$ is within $\eps$ of its true value with probability $1-\frac{\delta}{2^{n(n+1)}}$, and hence by the union bound, they are all within $\gamma$ of their true values with probability $1-\delta$. $N$ copies of the message corresponding to $\clP$ need to be sent in order to repeat it $N$ times, which gives the overall communication complexity of $\clQ_{\gamma,\delta}$.

Let $\eps \coloneqq 2^{-n/9}$ and $M \coloneqq \lceil \frac{64}{\epsilon} \rceil$. For each $i \in [M]$, we now define some `quantized' protocols $\clQ_i$ that work as follows on inputs $f, x$: 
\begin{itemize}
    \item Run $\clQ_{\frac{\eps}{4M},\frac{1}{5}}$ so that Bob obtains the list $\{P^{f,x}(z)\}$. 
    \item Bob outputs the value $z$ which satisfies $P^{f,x}(z) > \eps(1+ \frac{2i-1}{4M})$. If there are multiple such $z$, output the lexicographically smallest such value. If there is no such $z$, output $0^{n(n+1)}$.
\end{itemize}
Let us denote the most likely output of $\clQ_i(f,x)$ by $\hat{z}^i(f,x)$. We claim that there exists a subset of $(f,x)$ that has probability $2^{-o(n)}$ over $D_k^n\otimes U^n$, such that for each $(f,x)$ in the subset, with probability at least $\frac{31}{32}$ over $i\in [M]$, it holds that:
\begin{equation}\label{eq:hatz-prob}
\Pr[\clQ_i(f,x) = \hat{z}^i(f,x) \land \hat{z}^i(f,x) \in \HMn^n(f,x)] \geq \frac{4}{5}.
\end{equation}
The $2^{-o(n)}$ probability subset of $(f,x)$ is just going to be the good subset for which \eqref{eq:delta-contr} is true, and so we focus on such a pair. For this fixed $(f,x)$ pair, for at least $1-\frac{1}{32}$ fraction of $i\in [M]$, there does not exist a $z$ satisfying
\begin{equation}\label{eq:i-bad}
\left|\Pr[\clP(f,x)=z] - \eps\left(1+ \frac{2i-1}{4M}\right)\right| < \frac{\eps}{4M}.
\end{equation}
This is because such a $z$ must satisfy $\Pr[\clP(f,x)=z] > \eps$, so there can be at most $\frac{1}{\eps}$ of them. But since the step size in our quantization is $\frac{\eps}{4M}$, and the right-hand side of the above inequality is $\frac{\eps}{4M}$, each $z$ can satisfy the inequality with at most one $i$. Therefore, the fraction of $i$ for which there is some $z$ satisfying the inequality is $\frac{1}{\eps M} = \frac{1}{64}$. For $i$ for which \eqref{eq:i-bad} does not hold for any $z$, if $\clQ_{\frac{\eps}{4M},\frac{1}{5}}$ succeeds, it must find $\hat{z}^i$, and this must happen with probability at least $\frac{4}{5}$.
Moreover, each $\clQ_i$ outputs a unique $z$ which has probability more than $\eps$ of being output by $\clP$. Since $\clP$ is wrong on $\HMn^n(f,x)$ with probability at most $\frac{1}{32}$, the maximum number of $i$-s which output a $z$ that is not valid output for $\HMn^n(f,x)$ is at most $\frac{1}{32\eps}$, which is at most a $\frac{1}{64}$ fraction of all $i\in [M]$. Overall, for a good $(f,x)$, the probability over $i$ that $\clQ_i(f,x)$ outputs $\hat{z}^i(f,x)$ and that this output is in $\HMn^n(f,x)$ is at least $1-\frac{1}{64}-\frac{1}{64} = \frac{31}{32}$. Therefore, for these $i$, \eqref{eq:hatz-prob} holds.
We call this set of $i$ the good set this fixed good $(f,x)$.

Now for $j \in [n(n+1)]$, define protocols $\clQ_{ij}$ in which Bob just outputs the $j$-th bit of the output of $\clQ_i$. Further, for the set of $(f,x)$ and $i$ for which \eqref{eq:hatz-prob} holds, we also define the set of boolean functions $g_{ij}$ given by $g_{ij}(f,x)\coloneqq \hat{z}^i_j(f,x)$.
Then due to \eqref{eq:hatz-prob}, we have a set of one-way quantum protocols $\clQ_{ij}$ for $g_{ij}(f,x)$ that succeed with probability at least $\frac{4}{5}$. The communication complexity of $\clQ_{ij}$ is the same as that of $\clQ_{\eps/4M,1/5}$, which is $O\left(\frac{n^2p(n)}{\eps^4}\right)$.

Now by Theorem \ref{thm:QtoD}, there also exist deterministic one-way protocols $\clD_{ij}$ for $g_{ij}$ of complexity $O(\frac{n^4p(n)}{\eps^4})$. Note that for $(f,x)$ and $i$ for which \eqref{eq:hatz-prob} holds, $\clD_{ij}$ succeeds with probability $1$ for each $j\in [n(n+1)]$. For the bad $(f,x)$ and $i$ for which \eqref{eq:hatz-prob} does not hold, $\clD_{ij}$ outputs something, and we don't care what it is.

Using the protocols $\clD_{ij}$, we are going to give a randomized protocol $\clR$ for $\HMn^n$ that works with high probability over its internal randomness and the distribution $D_k^n\otimes U^n$ over $(f,x)$. By fixing the randomness of $\clR$, there is then also a deterministic protocol $\clD$ that works with high probability over this distribution $(f,x)$. $\clR$ works as follows:
\begin{itemize}
\item Alice samples $i\in [M]$ uniformly at random, and sends the message for $\clD_{ij}$ for each $j$, along with $i$, to Bob.
\item For each $j\in [n(n+1)]$, Bob computes his output bit for Alice's message in $\clD_{ij}$. He then outputs the concatenation of all the $n(n+1)$ bits.
\end{itemize}
If $(f,x)$ is good, and the $i$ sampled by Alice is good for this pair, then $\clR$ outputs $z\in \HMn^n(f,x)$ (since in this case, $\hat{z}^i(f,x)\in \HMn^n$, and $\clD_{ij}$ is correct for each $j$). We know that a $2^{-o(n)}$ probability subset of $(f,x)$ are good, and for each good pair, $1-\frac{1}{32}$ fraction of $i$ are good. Therefore, the overall success probability of $\clR$ over $D_k^n\otimes U^n$ is $2^{-o(n)}$. Fixing the randomness of $\clR$, we can also get a deterministic protocol $\clD$ which has the same success probability over the same distribution, and the same communication. The success probability of $\clD$ over $D^n_k\otimes U^n$ is $2^{-o(n)}$, and its communication complexity is $O(\frac{n^6p(n)}{\eps^4})$. Putting in the value of $\eps$, this is $o(2^{n/2})$ for any polynomial $p(n)$, which contradicts Theorem \ref{thm:HM-DP}. Therefore, \eqref{eq:delta-contr} cannot be true.
\end{proof}

Now we'll generalize Theorem \ref{thm:minent-1} to the $k$-receiver one-way communication setting. We'll use $f$ to denote Alice's input for this protocol as before, and $x^1\ldots x^k$ and $z^1\ldots z^k$ to denote the $k$ receivers' input and output variables. Moreover, we'll use $\bar{x}$ to refer to $x^1\ldots x^k$, $\bar{x}^{< \ell}$ for $x^1\ldots x^{\ell-1}$ and $\bar{x}^{\leq \ell}$ for $x^1\ldots x^\ell$; similar notation will be used for $z$. We'll use $\clP(f,\bar{x})$ to denote the (probabilistic) outcomes of all receivers in the $k$-receiver protocol $\clP$ on inputs $f$ and $\bar{x}$, and $\clP(f,x^\ell)$ to denote the outcome of the $\ell$-th receiver.

First, we'll prove Lemma \ref{lem:minent-chain} about the distribution of $z^\ell$ conditioned on $\bar{z}^{<\ell}$, on average over $\bar{x}$. We will then use this to prove in Lemma \ref{lem:minent-k} that there exists a $\bar{x}$ for which the min-entropy of the whole string $\bar{z}$ is high for a good fraction of $f$.

\begin{lemma}\label{lem:minent-chain}
Let $\clP$ be a $k$-receiver one-way communication protocol for $\HMn^n$ with perfect joint correctness. Then for any set of inputs $\bar{x}^{< \ell}$ and outputs $\bar{z}^{< \ell}$ for the first $\ell-1$ receivers, we have,
\[ \bbE_{f \sim D_{k'}^n, x^\ell \sim U^n}\left[\max_{z^\ell}\Pr\left[\clP(f,x^\ell) = z^\ell\middle|\clP(f,\bar{x}^{<\ell}) = \bar{z}^{< \ell}\right]\right] = 2^{-\Omega(n)},\]
for any $k'=\poly(n)$.
\end{lemma}

\begin{proof}
By Theorem \ref{thm:minent-1} for all 1-receiver protocols for $\HMn^n$, we have,
\[ \bbE_{f \sim D_k^n, x \sim U^n}\left[\max_z\Pr\left[\clP'(f,x)=z\right]\right] = 2^{-\Omega(n)}.\]
We can get a 1-receiver protocol $\clP'$ with perfect correctness for $\HMn^n$, whose output probabilities are exactly equal to the output probabilities of the $\ell$-th receiver of $\clP$, conditioned on the inputs and outputs of the first $\ell-1$ receivers being $\bar{x}^{< \ell}$ and $\bar{z}^{< \ell}$ respectively, in the following way. In $\clP'$, Alice prepares many copies of the whole message corresponding to the $k$ players for $\clP$, and performs the measurements of the first $\ell-1$ players corresponding to inputs $\bar{x}^{< \ell}$ on them until she gets outcomes $\bar{z}^{< \ell}$. Alice might need to prepare very many copies in order to do this, but this is costless. Alice then sends the $\ell$-th player's register to Bob, corresponding to a copy where she got the desired outcome. Since $\clP$ has perfect correctness, the $\ell$-th receiver's output is correct conditioned on any outcome of the first $\ell-1$ receivers, and therefore $\clP'$ has perfect correctness. If the number of qubits send to the $\ell$-th player in $\clP$ is polynomial, then the communication in $\clP'$ is also polynomial. This proves the lemma.
\end{proof}

The above lemma is the first place our proof begins to fail for protocols that are not zero-error. For a $k$-receiver protocol with bounded error, if we condition on arbitrary $\bar{x}^{< \ell}, \bar{z}^{< \ell}$, the $\ell$-th receiver's output $z^\ell$ is not necessarily correct with high probability. Therefore, we cannot get a protocol $\clP'$ that we can apply Theorem \ref{thm:minent-1} to, with this conditioning.

\begin{lemma}\label{lem:minent-k}
Let $\clP$ be a $k$-receiver one-way communication protocol for $\HMn^n$ with perfect joint correctness. Then there exists inputs $\bar{x}^* = x^{1,*}\ldots x^{k,*}$ to the $k$ players such that with probability $2^{-\Omega(kn)}$ over functions $f:[n]\times\{0,1\}^n \to \{0,1\}$, we have
\[ \max_{\bar{z}}\Pr[\clP(f,\bar{x})=\bar{z}] = 2^{-\Omega(kn)}.\]
\end{lemma}

\begin{proof}
The proof will be done by induction over each player's outcome. Let $t \geq 1$ be a constant such that by Lemma \ref{lem:minent-chain} we have for all $\bar{x}^{<\ell}, \bar{z}^{<\ell}$,
\[ \bbE_{f \sim D'^n_{k^2}, x^\ell \sim U^n}\left[\max_{z^\ell}\Pr\left[\clP(f,x^\ell) = z^\ell\middle|\clP(f,\bar{x}^{<\ell}) = \bar{z}^{< \ell}\right]\right] \leq 2^{-n/t},\]
for all product distributions $D'^n_{k^2}$, where the distribution on each $f_i$ has min-entropy at least $2^n - k^2$. Then at the end of the $\ell$-th induction step, we'll have inputs $\bar{x}^{\leq \ell,*}$ such that there is a distribution $D_{1,\ell}\otimes\ldots\otimes D_{n,\ell}$ such that each $D_{i,\ell}$ is a uniform distribution over at least $\left(\frac{1}{4t}\right)^{\ell-1}$ fraction of functions $f_i:\{0,1\}^n \to \{0,1\}$ (i.e., each $D_{i,\ell}$ has min-entropy $2^n - (\ell-1)\log(4t)$), and we have,
\[ \bbE_{f \sim D^n_\ell}\left[\max_{\bar{z}^{\leq \ell}}\Pr[\clP(f,\bar{x}^{\leq \ell, *})=\bar{z}^{\leq \ell}]\right] \leq 2^{-n\ell/2t},\]
where we are using $D^n_\ell$ to refer to $D_{1,\ell}\otimes\ldots\otimes D_{n,\ell}$ as before.

As before we'll use $D^n_\ell$ to refer to $D_{1,\ell}\otimes\ldots\otimes D_{n,\ell}$. In the base case of $\ell=1$, each $D_{i,1}$ is just the uniform distribution over $f_i$ for each $i$. We have by Lemma \ref{lem:minent-chain} that
\[ \bbE_{f\sim D^n_1,x^1\sim U^n}\left[\max_{z_1}\Pr[\clP(f,x^1)=z^1]\right] \leq 2^{-n/t},\]
for some $t \geq 1$. Therefore, we can find some $x^{1,*}$ for which the expectation over $f$ is at least $2^{-n/t}$.

For the induction step, we'll find a subdistribution of $D^n_{\ell}$ such that for every $f$ in the support of that subdistribution, $\max_{\bar{z}^{\leq \ell}}\Pr[\clP(f,\bar{x}^{\leq \ell,*}) = \bar{z}^{\leq \ell}] \leq 2^{-\ell n/2t+n/2t}$. We'll do this by applying Markov's inequality on each distribution $D_{i,\ell}$ separately (so that the overall distribution remains product). By Markov's inequality,
\[ \Pr_{f_1\sim D_{1,\ell}}\left[\bbE_{f_2\ldots f_n \sim D_{2,\ell}\otimes\ldots\otimes D_{n,\ell}}\max_{\bar{z}^{\leq \ell}}\Pr[\clP(f,\bar{x}^{\leq \ell, *}) = \bar{z}^{\leq \ell}] \geq 2^{-n\ell/2t + 1/2t}\right] \leq 2^{-1/2t}.\]
This means with probability at least $1-2^{-1/2t} \geq 1 - e^{-\ln(2)/2t} \geq \frac{\ln(2)}{2t} \geq \frac{1}{4t}$ over $f_1$ from $D_{1,\ell}$, the expectation over $f_2\ldots f_n$ is at most $2^{-n\ell/2t + 1/2t}$. We now set $D_{1,\ell+1}$ to be $D_{1,\ell}$ conditioned on this good $\frac{1}{4t}$ probability set. Since $D_{1,\ell}$ was a uniform distribution on at least $\left(\frac{1}{4t}\right)^{\ell-1}$ fraction of functions $f_1$, $D_{1,\ell+1}$ is a uniform distribution on $\left(\frac{1}{4t}\right)^{\ell}$ fraction of $f_1$. Doing a similar similar process for $f_2, \ldots, f_n$ in order we get distributions $D_{2,\ell+1}, \ldots, D_{n,\ell+1}$ such that for every $f_1\ldots f_n$ in the support of $D_{1,\ell+1}\otimes\ldots\otimes D_{n,\ell+1}$, we have $\max_{\bar{z}^{\leq \ell}}\Pr[\clP(f,\bar{x}^{\leq \ell,*}) = \bar{z}^{\leq \ell}] \leq 2^{-\ell n/2t+n/2t}$.

Now since $\ell \leq k = \poly(n)$, the distribution $D^n_{\ell+1} = D_{1,\ell+1}\otimes\ldots\otimes D_{n,\ell+1}$ satisfies the conditions of Lemma \ref{lem:minent-chain}. Therefore, by Lemma \ref{lem:minent-chain} we have for $\bar{x}^{\leq \ell, *}$ and every $\bar{z}^{\leq \ell}$,
\[ \bbE_{f\sim D^n_{\ell+1}, x^{\ell+1}\sim U^n}\left[\max_{z^{\ell+1}}\Pr\left[\clP(f,x^{\ell+1}) = z^{\ell+1}\middle|\clP(f,\bar{x}^{\leq \ell,*}) = \bar{z}^{\leq \ell}\right]\right] \leq 2^{-n/t}.\]
Since for every $f$ in the support of $D^n_{\ell+1}$, we have bounded the maximum output probability of $\bar{z}^{\leq \ell}$, we have,
\[ \bbE_{f\sim D^n_{\ell+1}, x^{\ell+1,} \sim U^n}\left[\max_{z^{\ell+1}}\Pr\left[\clP(f,\bar{x}^{\leq \ell,*}x^{\ell+1}) = z^{\leq(\ell+1)}\right]\right] \leq 2^{-\ell n/2t+n/2t}\cdot 2^{-n/t} \leq 2^{-(\ell+1) n/2t}.\]
Therefore, we can find an $x^{\ell+1,*}$ for which the expectation over $f$ is bounded.

Finally, applying Markov's inequality again at the end of $k$ steps, we get that with probability $\left(\frac{1}{4t}\right)^k$ over each of $f_1, \ldots, f_n$ (which means with probability at least $2^{-\Omega(kn)}$ over $f$),
\[ \max_{\bar{z}}\Pr[\clP(f,\bar{x}^*)=\bar{z}] = 2^{-kn/2t}. \qedhere\]
\end{proof}

\subsection{Diagonalization}
We are almost ready to prove the main result of this section (Theorem \ref{thm:comm-FEQP}): a family of $p(n)$-qubit states $\ket{\psi_f}$ cannot jointly be used by $k=O(p(n))$ many receivers to solve $\HMn^n$ (after the action of a cloner). The main components of the proof will be Lemma \ref{lem:minent-k} and the following well-known result about small pairwise inner products between a large number of vectors requiring high dimension.

\begin{theorem}[Theorem 8.3.2 in \cite{dW01}]\label{thm:finger}
Suppose that a family $\{\ket{h_x}\}_{x \in \{0,1\}^m}$ of $b$-qubit pure states satisfies $| \braket{h_x|h_y} | \leq \eps$ for any distinct $x,y\in \{0,1\}^m$, for $\eps > 2^{-m}$.  Then, $b = \Omega\left(\log\left(\frac{m}{\eps^2}\right)\right)$.
\end{theorem}

We'll also need the following lemma about the probability that a joint probability distribution of $k$ variables takes too few distinct values across all $k$ variables.
\begin{lemma}\label{lem:distinct}
Let $A^1\ldots A^k$ be random variables on $\{0,1\}^{Nk}$ such that $\max_{a^1\ldots a^k}\Pr[a^1\ldots a^k] \leq 2^{-Nk/t}$. Then for $ \log k \ll N$ we have,
\[ \Pr\left[A^1\ldots A^k \text{ takes at most $\frac{k}{2t}$ distinct values}\right] \leq 2^{-Nk/4t}.\]
\end{lemma}
\begin{proof}
The number of strings $a^1\ldots a^k$ with at most $\lfloor k/2tn\rfloor$ distinct values in $[2^{N}]$ is at most $2^{N k/2t}\cdot\left(\frac{k}{2t}\right)^{k- k/2t}$, because the first $k/2t$ $a^i$-s can take $2^{N}$ values each, and the rest of the $a^i$-s each have to take one of the at most $k/2t$ values previously taken. Since each string has probability at most $2^{-Nk/t}$, we have,
\begin{align*}
\Pr_{A^1\ldots A^k}\left[a^1\ldots a^k \text{ has at most $\frac{k}{2t}$ distinct values in $[2^{N}]$}\right] & \leq 2^{N k/2t}\cdot\left(\frac{k}{2t}\right)^{k- k/2t}\cdot 2^{-Nk/t} \\
 & \leq 2^{Nk/2t}\cdot 2^{k\log k}\cdot 2^{-Nk/t} \leq 2^{-Nk/4t}
\end{align*}
for $\log k \ll N$.
\end{proof}

\begin{theorem}\label{thm:comm-FEQP}
Suppose $\{\ket{\psi_f}\}_f$ is a family of $p(n)$-qubit states for each $f:[n]\times\{0,1\}^n \to \{0,1\}$. Then there exists a $k=O(p(n))$ such that for any algorithm $\clA$ that takes in $\ket{\psi_f}$ and outputs $\ket{\psi'_f}$ (without having access to $f$) on registers $B_1\ldots B_k$, and any measurements $k$ receivers $\clB_1,\ldots, \clB_k$ perform on registers $B_1, \ldots, B_k$ respectively, $\clB_1, \ldots, \clB_k$ cannot jointly solve $\HMn^n$ with perfect correctness with the one-way message states from Alice being $\ket{\psi'_f}$.
\end{theorem}

\begin{proof}
Suppose for the sake of contradiction that $\{\ket{\psi'_f}\}$ is a valid family of message states for a $k$-receiver one-way protocol for $\HMn^n$ with perfect correctness, for $k=O(p(n))$ (we'll specify the constant in the big $O$ later). We can assume without loss of generality that $\ket{\psi_f}$ and $\ket{\psi'_f}$ are related by an isometry, and since the isometry does not depend on $f$, we have, $\brakett{\psi_f}{\psi_{f'}} = \brakett{\psi'_f}{\psi'_{f'}}$ for any $f, f'$.

Since $\clB_1, \ldots, \clB_k$ can jointly solve $\HMn^n$ with perfect correctness, using Lemma \ref{lem:minent-k}, for some input string $\bar{x}^*$ to the Bobs and for $(4t)^{-kn}$ fraction of functions $f:[n]\times\{0,1\}^n\to\{0,1\}$, we have for all output strings $\bar{z}$,
\[
\Pr[\clP(f,\bar{x}^*)=\bar{z}] \leq 2^{-nk/t},
\]
for some constant $t \geq 1$. Let $\tilde{F}$ denote the $(4t)^{-kn}$ fraction of functions for which this holds. A single player's output $z^\ell$ for $\HMn^n$ consists of $((y^\ell_1,f_1(y^\ell_1)\oplus f_1(y^\ell_1\oplus x^\ell_1)),\ldots,(y^\ell_n,f_n(y^\ell_n)\oplus f_n(y^\ell_n\oplus x^\ell_n)))$. Since the protocol is zero-error, for a fixed $x^\ell$ and $f$, $f_1(y^\ell_1)\oplus f_1(y^\ell_1\oplus x^\ell_1), \ldots, f_n(y^\ell_n)\oplus f_n(y^\ell_n\oplus x^\ell_n)$ are fixed once $y^\ell = y^\ell_1\ldots y^\ell_n$ is fixed. So the distribution of $z^\ell$ is really the distribution of $y^\ell$, and we have for $f \in \tilde{F}$ and all $\bar{y}$,
\begin{equation}\label{eq:prob-bound}
\Pr[\clP(f,\bar{x}^*)=\bar{y}] \leq 2^{-nk/t}
\end{equation}

Note that for the $\HMn$ problem, outputting $(y_i,f_i(y_i)\oplus f_i(y_i\oplus x_i))$ and $(y_i\oplus x_i, f_i(y_i)\oplus f_i(y_i\oplus x_i))$ are equivalent on input $x_i$. Therefore, we'll assume there is a single probability of each receiver outputting one of these, corresponding to their input. Moreover, we'll assume the measurement basis state corresponding to the output $(\{y_i,y_i\oplus x_i\}, 0)$ is $\frac{1}{\sqrt{2}}(\ket{y_i}+\ket{y_i\oplus x_i})$, and that corresponding to $(\{y_i, y_i\oplus x_i\}, 1)$ is $\frac{1}{\sqrt{2}}(\ket{y_i}-\ket{y_i\oplus x_i})$. Note that the states $\left\{\frac{\ket{y_i}\pm\ket{y_i\oplus x_i}}{\sqrt{2}}\right\}_{y_i < y_i\oplus x_i}$ (where we are enforcing $y_i < y_i\oplus x_i$ according to some lexicographic ordering so as not to overcount) form a full orthonormal basis for any $x_i$. Extending this to $\HMn^n$, the measurement basis state corresponding to output $(y_1\ldots y_n,(f_1(y_1)\oplus f_1(y_1\oplus x_1),\ldots, f_n(y_n)\oplus f_n(y_n\oplus x_n)))$ (and also all $2^n$ strings formed by XORing some of the $y_i$-s with $x_i$-s) is $\bigotimes_{i=1}^n\frac{\ket{y_i}+(-1)^{f_i(y_i)\oplus f_i(y_i\oplus x_i)}\ket{y_i\oplus x_i}}{\sqrt{2}}$. With some abuse of notation, we'll use $\frac{\ket{y}+(-1)^f\ket{y\oplus x}}{2^{n/2}}$ to denote this state, and use $\frac{\ket{y}\pm\ket{y\oplus x}}{2^{n/2}}$ to refer to the measurement basis used when the input is $x$.

Our assumption for the receivers $\clB_1, \ldots, \clB_k$ is that each receiver applies some unitary depending on their input $x^\ell$ to the state they receive from $\clA$, before measuring in the $\frac{\ket{y}\pm\ket{y\oplus x^\ell}}{2^{n/2}}$ basis. Since the protocol has perfect correctness, $\clB_\ell$'s measurement register just before the final measurement can have non-zero amplitude on the state $\frac{1}{\sqrt{2}}(\ket{y}+(-1)^{f}\ket{y\oplus x^\ell})$, but not on any of the states orthogonal to it in the basis, when their input in $x^\ell$.

We'll denote the measurement register of the $\ell$-th player by $M_\ell$, and their other registers by $N_\ell$. Let us call the state just before the final measurement of each player for input $\bar{x}$, $\ket{\phi_{f,\bar{x}}}$. By the previous discussion, the state must look like,
\begin{align*}
\ket{\phi_{f,\bar{x}}} = \sum_{\bar{y}< \bar{y}\oplus \bar{x}}\alpha_{f,\bar{x}}(\bar{y})\left(\frac{\ket{y^1}+(-1)^{f}\ket{y^1\oplus x^1}}{2^{n/2}}\right)_{M_1}\otimes\ldots \otimes\left(\frac{\ket{y^k}+(-1)^{f}\ket{y^k\oplus x^k}}{2^{n/2}}\right)_{M_k}\otimes \ket{\phi'_{f,\bar{x},\bar{y}}}_{N_1\ldots N_k}.
\end{align*}
In the above expression, the
 coefficient $\alpha_{f,\bar{x}}(\bar{y})$ is the amplitude for outcome $\bar{y}$ corresponding to inputs $f$ and $\bar{x}$. Therefore, for $f \in \tilde{F}$, and $\bar{x} = \bar{x}^*$, the distribution given by $|\alpha_{f,\bar{x}}(\bar{y})|^2$ must satisfy \eqref{eq:prob-bound}.

We can express the inner product between two of the states $\ket{\phi_{f,\bar{x}}}$ and $\ket{\phi_{f',\bar{x}}}$ as
\begin{align}\label{eq:ff'-prod}
|\brakett{\phi_{f,\bar{x}}}{\phi_{f',\bar{x}}}| & = \left|\sum_{\bar{y}< \bar{y}\oplus \bar{x}}\alpha_{f,\bar{x}}(\bar{y})\alpha_{f',\bar{x}}(\bar{y})I_{f,f',y^1,x^{1}}\ldots I_{f,f',y^k,x^{k}}\brakett{\phi'_{f,\bar{x},\bar{y}}}{\phi'_{f',\bar{x},\bar{y}}}\right| \nonumber \\
 & \leq \left|\sum_{\bar{y}< \bar{y}\oplus \bar{x}}\alpha_{f,\bar{x}}(\bar{y})\alpha_{f',\bar{x}}(\bar{y}) I_{f,f',y^1,x^{1}}\ldots I_{f,f',y^k,x^{k}}\right|,
\end{align}
where $ I_{f,f',y^\ell,x^{\ell}}$ is the indicator function for $f(y^\ell)\oplus f(y^\ell\oplus x^\ell) = f'(y^\ell)\oplus f'(y^\ell\oplus x^\ell)$ (where $f(y^\ell)\oplus f(y^\ell\oplus x^\ell)$ is the string $(f_1(y^\ell_1)\oplus f_1(y^\ell_1\oplus x^\ell_1)), \ldots, (f_n(y^\ell_n)\oplus f_n(y^\ell_n\oplus x^\ell_n))$), since the states $\frac{\ket{y^\ell}+(-1)^{f}\ket{y^\ell\oplus x^\ell}}{2^{n/2}}$ and $\frac{\ket{y^\ell}+(-1)^{f'}\ket{y^\ell\oplus x^\ell}}{2^{n/2}}$ are orthogonal otherwise. For $f'\in \tilde{F}$, we want to upper bound the probability over uniformly random $f$ that the inner product $\brakett{\phi_{f,\bar{x}^*}}{\phi_{f',\bar{x}^*}}$ is high. For this it will be useful to upper bound the expectation of a higher order moment of it.

For some integer $c>1$ to be determined later, consider
\begin{align}
|\brakett{\phi_{f,\bar{x}^*}}{\phi_{f',\bar{x}^*}}|^c & \leq \left|\sum_{\bar{y}< \bar{y}\oplus \bar{x}^*}\alpha_{f,\bar{x}^*}(\bar{y})\alpha_{f',\bar{x}^*}(\bar{y}) I_{f,f',y^1,x^{1,*}}\ldots I_{f,f',y^k,x^{k,*}}\right|^c \nonumber \\
 & \leq \left(\sum_{\bar{y}< \bar{y}\oplus \bar{x}^*}|\alpha_{f,\bar{x}^*}(\bar{y})|^2\sum_{\bar{y}' < \bar{y}'\oplus \bar{x}^*}|\alpha_{f',\bar{x}^*}(\bar{y}')|^2 I_{f,f',y'^1,x^{1,*}}\ldots I_{f,f',y'^k,x^{k,*}}\right)^c \nonumber \\
 & = \left(\sum_{\bar{y}' < \bar{y}'\oplus \bar{x}^*}|\alpha_{f',\bar{x}^*}(\bar{y}')|^2 I_{f,f',y'^1,x^{1,*}}\ldots I_{f,f',y'^k,x^{k,*}}\right)^c \nonumber
 \end{align}
The second inequality above uses the Cauchy-Schwarz inequality, and the first equality is due to the observation that $\sum_{\bar{y}< \bar{y}\oplus \bar{x}^*}|\alpha_{f,\bar{x}^*}(\bar{y})|^2 = 1$ for all $f$. We can expand the $c$-th power in the last expression by doing $c$ sums over $\bar{y}'$-s, with each term in the sum being the product of $|\alpha_{f',\bar{x}^*}(\bar{y}')|^2$ and the indicator variables for each of the $c$ $\bar{y}'$-s. Alternatively, this can be written as a sum over $\tilde{y}$ where $\tilde{y}$ is a tuple containing $c$ $\bar{y}'$-s, i.e., a string in $\{0,1\}^{n^2ck}$. With some abuse of notation, we use $y^1, \ldots, y^{ck}$ to refer to each of the $n^2$-bit blocks of a $\tilde{y}$. We'll also use $\tilde{x}^*$ to refer to $c$ copies of $\bar{x}^*$, and $x^1, \ldots, x^{ck}$ to refer to its $n^2$-bit blocks, though these $x^\ell$-s repeat after intervals of $k$. Then we have,
\[ |\brakett{\phi_{f,\bar{x}^*}}{\phi_{f',\bar{x}^*}}|^c \leq \sum_{\tilde{y} < \tilde{y}\oplus \tilde{x}^*}p_{f',\bar{x}^*}(\tilde{y}) I_{f,f'y^1,x^{1,*}}\ldots  I_{f,f',y^{ck}x^{ck,*}},\]
where $p_{f',\bar{x}^*}(\tilde{y})$ is the probability distribution given by the product of $c$ copies of the distribution $|\alpha_{f',\bar{x}^*}(\bar{y})|^2$. For $f'\in \tilde{F}$, due to \eqref{eq:prob-bound} we have, $\max_{\tilde{y}}p_{f',\bar{x}^*}(\tilde{y}) \leq 2^{-nck/t}$. Since $\tilde{y}$ consists of $ck$ random variables over $\{0,1\}^{n^2}$, we have by Lemma \ref{lem:distinct}, that with probability at least $1-2^{-nck/4t}$ over $p_{f,\bar{x}^*}(\tilde{y})$, $\tilde{y} = y^1\ldots y^{ck}$ take at least $\frac{ck}{2tn}$ distinct values in $\{0,1\}^{n^2}$ (as long as $\log (ck) \ll n^2$ --- we'll later pick $c$ such that this is true). Fixing an $f'\in\tilde{F}$, let $b^1\ldots b^{ck}$ be the string where $b^\ell=f'(y^\ell)\oplus f'(y^\ell\oplus x^\ell)$ is an $n$-bit string. For a string $\tilde{y}$ that consists of at least $ck/2tn$ distinct values, the indicator variables for these values are independent, and we have,
\begin{align*}
\bbE_f\left[ I_{f,f',y^1,x^{1,*}}\ldots I_{f,f',y^{ck},x^{ck,*}}\right] & = \Pr_f\left[(f(y^1)\oplus f(y^1\oplus x^{1,*}) = b^1)\land\ldots\land(f(y^{ck})\oplus f(y^{ck}\oplus x^{ck,*}))\right] \\
& \leq (2^{-n})^{ck/2tn} = 2^{-ck/2t},
\end{align*}
since the probability of $(f(y^1)\oplus f(y^\ell\oplus x^{\ell,*}) = b^\ell)$ for each distinct $y^\ell$ is $2^{-n}$. Fixing an $f'\in \tilde{F}$, let $\text{Good}$ denote the subset of $\tilde{y}$-s which contain at least $ck/2tn$ distinct values. Then we have,
 \begin{align*}
\bbE_{f}\left[|\brakett{\phi_{f,\bar{x}^*}}{\phi_{f',\bar{x}^*}}|^c\right] & \leq \sum_{\tilde{y} \in \text{Good}}p_{f,\bar{x}^*}(\tilde{y})\bbE_f\left[ I^1_{f,f',y^1,x^{1,*}}\ldots I^k_{f,f',y^{ck},x^{ck,*}}\right] + \sum_{\tilde{y}\not\in\text{Good}} p_{f,\bar{x}^*}(\tilde{y})\cdot 1 \\
 & \leq \sum_{\tilde{y} \in \text{Good}}p_{f,\bar{x}^*}(\tilde{y})\cdot 2^{-ck/2t} + 2^{-nck/4t} \\
 & \leq 2^{-ck/4t}.
\end{align*}

By Markov's inequality, we then have for each $f'\in \tilde{F}$,
\begin{align*}
\Pr_f\left[|\brakett{\phi_{f,\bar{x}^*}}{\phi_{f',\bar{x}^*}}| > 2^{-k/8t}\right] & = \Pr_f\left[|\brakett{\phi_{f,\bar{x}^*}}{\phi_{f',\bar{x}^*}}|^c > 2^{-ck/8t}\right] \leq 2^{-ck/4t}\cdot 2^{ck/8t} \leq 2^{-ck/8t}.
\end{align*}
Note that the above probability is over uniformly random $f$, and now we want to calculate the probability over $f\in \tilde{F}$. Putting in $c=32nt\log(t)$, we have that a $2^{-4kn\log t}$ fraction of uniform $f$ have small inner product for each $f'\in \tilde{F}$. Since the probability of uniform $f$ being in $\tilde{F}$ is at least $2^{-2kn\log t}$, we have for each $f'$, at least $2^{-2kn\log t} - 2^{-4kn\log t}$ fraction of $f$ are both in $\tilde{F}$ and have small inner product. This means,
\[ \Pr_{f, f'\sim \tilde{F}}\left[|\brakett{\phi_{f,\bar{x}^*}}{\phi_{f',\bar{x}^*}}| \leq 2^{-k/8t}\right] \geq \frac{2^{-2kn\log t} - 2^{-4kn\log t}}{2^{-2kn\log t}} \geq 1 - 2^{-2kn\log t}.\]

Now consider a graph whose vertices are labelled by $f \in \tilde{F}$ where two functions have an edge between them if the corresponding inner product is at most $2^{-k/8t}$. By the above equation, this graph has at least $\left(1-2^{-2kn\log t}\right)\cdot \frac{|\tilde{F}|}{2}$ edges (we have divided by $2$ so as not to count $(f,f')$ and $(f',f)$ separately). Therefore, by Turan's theorem, this graph contains a complete subgraph of size at least $2^{kn\log t}$. This means we have a set $F^*$ of $2^{kn\log t}$ functions such that for every $f \neq f' \in F^*$ we have, $|\brakett{\phi_{f,\bar{x}^*}}{\phi_{f',\bar{x}^*}}| \leq 2^{-k/8t}$. Since $\ket{\phi_{f,\bar{x}*}}$ and $\ket{\phi_{f',\bar{x}^*}}$ are obtained by applying the same unitary depending on $\bar{x}^*$ to $\ket{\psi_f}$ and $\ket{\psi_{f'}}$, we have for every $f\neq f'\in F^*$,
\[ |\brakett{\psi_f}{\psi_{f'}}| = |\brakett{\psi'_f}{\psi'_{f'}}| = |\brakett{\phi_{f,\bar{x}^*}}{\phi_{f',\bar{x}^*}}| \leq 2^{-k/8t}.\]
By Theorem \ref{thm:finger}, the number of qubits in the states $\ket{\psi_f}$ must then be at least $C(\log(kn\log t) + k/4t)$. Recalling that the states $\ket{\psi_f}$ have $p(n)$ qubits by hypothesis, specifying $k(n) = Cp(n)/4t$, this is a contradiction.
\end{proof}

Finally, we'll do a diagonalization argument similar to Theorems \ref{thm:diag-suQMA} and \ref{thm:diag-BQPsuq}.
\begin{theorem}\label{thm:diag-FEQPsuq}
There exists a relation $R \in \FEQPsuq$.
\end{theorem}

\begin{proof}
We once again consider tuples $s = (p,\clA,\clB_1,\ldots, \clB_k)$ where $p(n)$ is a polynomial growth rate of the advice, $k=O(p(n))$ is as given by Theorem \ref{thm:comm-FEQP} and $\clA, \clB_1, \ldots, \clB_k$ are algorithms. By Theorem \ref{thm:comm-FEQP}, $\clB_1, \ldots, \clB_k$ cannot jointly solve $\HMn^n$ with perfect correctness with the states provided by $\clA$ by acting on any family of $p(n)$-qubit advice states. Therefore, there exists at least one $f$ for which one of the $\clB_i$-s fails. We'll fix the relation for the input length $N(s)$ corresponding to the tuple $s$ to be $\HMn^n(f,\cdot)$ for this $f$. The problem is obviously in $\FEQPq$ as well, using Alice's message state corresponding to this $f$ as the advice state.
\end{proof}

\section*{Acknowledgements}
R.C. is supported by the National Research Foundation, Singapore, under its NRF Fellowship programme, award no. NRF-NRFF14-2022-0010. S.K. is supported by the Natural Sciences and Engineering Research Council of Canada (NSERC) Discovery Grants Program, and Fujitsu Labs America. S.P.\ is supported by US Department of Energy (grant no DE-SC0023179) and partially supported by US National Science Foundation (award no 1954311).
Part of this work was conducted while S.P.\ was visiting the Simons Institute for the Theory of Computing, supported by DOE QSA (grant no FP00010905).

We thank Anne Broadbent, Martti Karvonen and Ernest Tan for helpful discussions, and anonymous reviewers for their comments on improving the manuscript.

\bibliographystyle{plain}
\bibliography{Main}

\end{document}

%% file: Macros.tex
\usepackage[normalem]{ulem}
\newcommand{\dlt}{\bgroup\markoverwith{\textcolor{red}{\rule[0.5ex]{1pt}{0.6pt}}}\ULon}

\newcommand{\bits}{\ensuremath{\{0,1\}}\xspace}

\newcommand{\negl}{\ensuremath{\mathsf{negl}}\xspace}

\newcommand{\poly}{\ensuremath{\mathsf{poly}}\xspace}

\renewcommand{\tilde}{\widetilde}

\newcommand{\I}{\mathrm{I}}
\renewcommand{\H}{\mathrm{H}}

\renewcommand{\epsilon}{\varepsilon}
